\documentclass[ecta,nameyear,draft]{econsocart}

\RequirePackage[colorlinks,citecolor=blue,linkcolor=blue,urlcolor=blue,pagebackref]{hyperref}
\usepackage{booktabs}
\usepackage{tikz}
\usepackage{multirow}
\startlocaldefs

\numberwithin{equation}{section}
\theoremstyle{plain}
\newtheorem{prop}{Proposition}

\usepackage[makeroom]{cancel}
\newtheorem{assumption}{Assumption}
\newtheorem{theorem}{Theorem}

\theoremstyle{remark}

\newtheorem{remark}{Remark}
\usepackage{comment}
\endlocaldefs

\begin{document}
\begin{frontmatter}

\title{Semi-Nonparametric Models of Multidimensional Matching: an Optimal Transport Approach}
\runtitle{Semi-Nonparametric Multidimensional Matching}

\begin{aug}
% use \particle for den|der|de|van|von (only lc!)
% [id=?,addressref=?,corref]{\fnms{}~\snm{}\ead[label=e?]{}\thanksref{}}
%
%% e-mail is mandatory for each author
%
%%% initials in fnms (if any) with spaces
%
\author[id=au1,addressref={add1}]{\fnms{Dongwoo}~\snm{Kim}\ead[label=e1]{dongwook@sfu.ca, dongwoo\_kim@korea.ac.kr}}
\author[id=au2,addressref={add2}]{\fnms{Young Jun}~\snm{Lee}\ead[label=e2]{youngjunlee@hanyang.ac.kr}}

%%%%%%%%%%%%%%%%%%%%%%%%%%%%%%%%%%%%%%%%%%%%%%
%% Addresses  %%
%%%%%%%%%%%%%%%%%%%%%%%%%%%%%%%%%%%%%%%%%%%%%%
\address[id=add1]{%
\orgname{Simon Fraser University and Korea University}}

\address[id=add2]{%
\orgname{Hanyang University}}
\end{aug}

%% Put support info here. Reminder: do not thank the handling coeditor anonymously or by name
\support{This paper is based on the third chapter of Lee's doctoral dissertation at University College London. We thank Dennis Kristensen, Krishna Pendakur, Minchul Shin, Martin Weidner, and Daniel Wilhelm for their helpful suggestions. We also benefited greatly from comments by seminar and conference participants at UCL, Bocconi University, KIEP, University of Calgary, Northwestern University, Kyung Hee University, University of Seoul, Hanyang University, KAEA virtual seminar, IAAE 2024, KERIC 2024, ESWC 2025, and EcoSta 2025. Kim gratefully acknowledges support from the Social Sciences and Humanities Research Council of Canada under the Insight Grant (435-2024-0322) and Korea University Grants (K2528561, K2528721, K2614081, K2613941). Lee gratefully acknowledges support from the Italian Ministry of University and Research under PRIN 2017 (2017TMFPSH). All errors are our own.}

\begin{abstract}
This paper develops a set of empirically tractable and flexible sieve estimators for semi-nonparametric multidimensional matching models with transferable utility, focusing on worker-job matching. We generalize the parametric quadratic-Gaussian framework employed by Bojilov and Galichon (2016) and Lindenlaub (2017), which relies on joint normality of observed characteristics. We allow unrestricted distributions of characteristics and show identification of the production technology and the equilibrium wage and matching functions using optimal transport theory. Given identification, we propose efficient, consistent, and asymptotically normal sieve estimators. We revisit Lindenlaub's empirical application and show that, between 1990 and 2010, the U.S. economy experienced much larger technological progress favoring cognitive abilities than the original findings suggest. Furthermore, our flexible model specifications provide a significantly better fit for patterns in the evolution of wage inequality.
\vspace{1em}
\end{abstract}

\end{frontmatter}
%%%%%%%%%%%%%%%%%%%%%%%%%%%%%%%%%%%%%%%%%%%%%%%%%%%%%%%%%%%%%%%%%%%%%%%%%
%%%% Main text entry area:
%%%%%%%%%%%%%%%%%%%%%%%%%%%%%%%%%%%%%%%%%%%%%%%%%%%%%%%%%%%%%%%%%%%%%%%%%
\newpage

\section{Introduction}\label{sec: Matchingintro}

In two-sided markets, agents form optimal matches based on their preferences and characteristics, generating a shareable joint surplus. Matching models are widely used to analyze these dynamics in settings such as the labor market (workers and jobs) and the marriage market (spouses). To keep the analysis tractable, empirical work often relies on low-dimensional heterogeneity, e.g., focusing on one salient attribute or aggregating multiple characteristics into a scalar index.\footnote{Assortative spousal matching on income, wages, education, risk aversion, and preference for childbearing are investigated by \cite{becker1991treatise}, \cite{grossbard1993theory}, \cite{pencavel1998assortative}, \cite{choo2006marries}, \cite{chiappori2016matching, legros2007beauty} and \cite{chiappori2008birth} among many others. \cite{becker1973} and \cite{chiappori2012fatter} investigate spousal matching that hinges on ``ability indices''.} While such scalar-index approaches implicitly encode multidimensional characteristics, they restrict the matching mechanism to operate along a single dimension. In the labor market, for instance, workers develop specialized skills (cognitive for mathematicians and manual for gymnasts), and the single index model cannot capture this specialization. These considerations motivate matching models that accommodate multidimensional heterogeneity more directly, a goal pursued in a growing literature.\footnote{Studies such as \cite{willis1979education} and \cite{papageorgiou2014learning} favor the multidimensional setup over a single index model in the labor market context. Spousal choices are also based on a variety of attributes as shown in \cite{becker1991treatise}, \cite{weiss1997match}, \cite{qian1998changes}, \cite{silventoinen2003assortative}, \cite{hitsch2010matching}, and \cite{oreffice2010anthropometry}.}

Following the seminal work of \cite{tinbergen1956theory}, recent contributions have advanced the analysis of matching with multiple continuous attributes. These frameworks typically achieve tractability by imposing parametric restrictions on the distribution of observable characteristics and/or on the structure of the surplus. In particular, \cite{bojilov2016} and \cite{lindenlaub2017} develop parametric matching models that assume joint normality of attributes and a bilinear surplus, which yields closed-form expressions for the equilibrium assignment. However, the normality assumption can be consequential. If the true joint distribution deviates from Gaussianity, the implied assignment patterns may be distorted. In empirical applications, this approach often requires additional steps to align the data with the normality assumption, for instance, transforming each marginal to standard normal. Such transformations can alter the interpretation of the underlying attributes, and the resulting joint distribution need not be well approximated by a multivariate normal. Consequently, the estimated assignment mechanism may not accurately reflect the true matching process. 

In this paper, we address these challenges by developing a set of empirically tractable and flexible sieve estimators for semi-nonparametric multidimensional matching models with transferable utility in the context of worker-job matching. We generalize the quadratic-Gaussian structure employed in \cite{bojilov2016} and \cite{lindenlaub2017} by accommodating arbitrary distributions of attributes while maintaining the bilinear surplus. In this context, each worker possesses distinct skills, and each job requires specific skills to produce output according to the production technology. The social planner's problem is to optimally assign workers to jobs to maximize total output in the economy, which can be formulated as a Monge-Kantorovich optimal transport problem. Unique solutions for the equilibrium assignment and wage functions are derived using an optimal transport approach \citep{villani2003,villani2008, dephilippis2014}.\footnote{Applications of optimal transport have proven successful in multiple fields of economics (e.g., \cite{ekeland2010notes}, \cite{chiappori2010hedonic}, \cite{chiong2016duality}, \cite{lindenlaub2017}, \cite{galichon2022cupid}, among others). \cite{galichon2017survey} offers a comprehensive survey of the literature.} We then propose an econometric framework to apply the theoretical model to real-world data. Following \cite{lindenlaub2017}, we introduce error terms to bridge the gap between the deterministic equilibrium predictions and observed data. Unlike \cite{lindenlaub2017}, however, we do not require the errors to be normally distributed or uncorrelated.

We estimate the production technology, equilibrium assignment, and wage functions using the semiparametric M-estimation techniques proposed by \cite{ai2003} and \cite{chen2007}. To our knowledge, this paper is the first in the literature to introduce semiparametric M-estimators to multidimensional matching models. Depending on the assumptions on the error terms, the model is estimated by sieve maximum likelihood (SML), least squares (SLS), or generalized least squares (SGLS). These estimators are efficient, asymptotically normal, and easy to implement. Our estimators perform well in extensive simulation experiments for a wide class of data-generating processes. Although we focus on worker-job matching, our method can be broadly applied to other matching problems, such as couple matching in the marriage market.

We revisit \cite{lindenlaub2017}'s worker-job data from the U.S. and estimate production technology, equilibrium assignment, and wage functions to investigate the technological shift and its effects on wage inequality between 1990 and 2010. Our results show greater technological progress favoring cognitive skills than Lindenlaub's estimates suggest. Furthermore, the greater flexibility of our models provides much greater explanatory power for the evolution of wage inequality, particularly the `\textit{wage polarization}' phenomenon featuring stronger wage growth in the bottom and upper tails of the wage distribution relative to the median. The Gaussian model fails to predict wage polarization because, on top of misspecification bias, it restricts the equilibrium wage function to a quadratic form. 

This paper is organized as follows. The remainder of this section discusses the related literature. Section \ref{sec: Matchingmodel} introduces the multidimensional matching model with a bilinear surplus and characterizes the equilibrium using optimal transport theory. Section \ref{sec: MatchingModelID} proposes the econometric matching models and establishes identification. Section \ref{sec:MatchingEst} presents the sieve estimators. Section \ref{sec: asymptotics} derives the asymptotic properties of our sieve GLS estimator. Section \ref{sec: Simul} conducts simulation experiments. Section \ref{sec: Matchingemp} revisits \citet{lindenlaub2017}'s empirical analysis. Section \ref{sec: conclusion} concludes. Technical proofs and additional theoretical details are provided in the appendix.

\subsection{Related literature}

\citet{choo2006marries} (CS henceforth) introduced an empirical transferable utility (TU) model for matching with discrete types and multidimensional unobserved heterogeneity.\footnote{See \citet{galichon2019} for the imperfectly transferable utility model with unobserved heterogeneity.} Their framework assumes that i.i.d.\ unobserved heterogeneity follows the extreme value type I (Gumbel) distribution, under which the systemic match surplus is identified by logit formulae. Several extensions remain within this discrete-type framework: \citet{galichon2022cupid} generalize the distributional assumption on unobservables beyond the logit specification and establish identification of the social surplus, while \cite{gualdani2023partial} show partial identification of the systemic surplus under nonparametric assumptions on the error distribution.

\citet{dupuy2014} extend the CS framework to continuous types while maintaining the i.i.d.\ Gumbel structure for unobserved heterogeneity, formulating the matching problem as an entropy-regularized optimal transport problem. \citet{bojilov2016} specialize the \citeauthor{dupuy2014} framework by imposing the quadratic-Gaussian structure---a bilinear surplus function combined with jointly normal attributes---and derive closed-form solutions for the equilibrium matching distribution. Taking a fundamentally different approach, \citet{lindenlaub2017} formulates a deterministic matching model without unobserved heterogeneity in the surplus, so the equilibrium assignment is governed by a standard Monge-Kantorovich optimal transport problem.\footnote{Alternatively, \citet{lise2020} consider a search-theoretic model in which workers are matched to firms in a dynamic setup.} Although \citeauthor{lindenlaub2017} also employs the quadratic-Gaussian structure to obtain closed-form solutions, the absence of unobserved heterogeneity yields a purely deterministic assignment rather than the stochastic matching patterns in the CS--\citeauthor{dupuy2014}--\citeauthor{bojilov2016} tradition.

All these approaches share the separability assumption that unobserved heterogeneity (or, in our case, measurement error) does not interact with observed characteristics in generating the surplus. However, the nature of unobservables differs fundamentally across frameworks. In the CS tradition (including \citeauthor{dupuy2014} and \citeauthor{bojilov2016}), unobserved heterogeneity enters the surplus function directly and determines stochastic matching patterns; the systemic surplus is identified from matching patterns---i.e., the joint distribution of matched characteristics---via logit inversion or its generalizations, and individual-level wage data are not required. In contrast, our framework and \citeauthor{lindenlaub2017}'s introduce measurement errors that are additive and separable, capturing discrepancies between the deterministic equilibrium predictions and observed data. Our econometric model explicitly exploits observed wages alongside matching patterns: the equilibrium wage function enters as a directly estimable object, which strengthens identification and estimation of the technology parameters.

Our paper extends \citeauthor{lindenlaub2017}'s deterministic matching framework by dispensing with distributional assumptions on the characteristics, thereby offering a more flexible and robust econometric framework for multidimensional matching. We propose efficient econometric procedures that jointly estimate both finite-dimensional parameters and infinite-dimensional functions using conditional moments implied by the model equilibrium, leveraging the sieve M-estimation literature.\footnote{\citet{shen1997} establishes asymptotic properties of smooth functionals of sieve MLE. \citet{newey2003}, \citet{ai2003, ai2007}, and \citet{blundell2007} propose efficient sieve IV and sieve minimum distance (SMD) estimators. \citet{chen2009} further show that the SMD estimator under proper penalization is consistent and efficient when residuals are potentially nonsmooth. \citet{chen2007} provides an extensive overview of sieve estimation of semi-nonparametric models.} In particular, we employ the sieve GLS estimator proposed in \citet{chen2007} for our most flexible model specification. This estimator is efficient and computationally simpler than the SMD estimator.
 
We establish the convergence, efficiency, and asymptotic normality of our sieve estimators relying on the smoothness of unknown nonparametric components (optimal transport maps). This smoothness condition can be verified by applying the results in the mathematical literature on optimal transport maps. \citet{caffarelli1992CPAM,caffarelli1992JAMS,caffarelli1996} show the smoothness of transport maps when the distributions of characteristics on both sides are compactly supported. \citet{cordero-erausquin2019} further extend the earlier results to cases where the distributions may have unbounded supports. The degree of smoothness of a transport map depends on how smooth the densities are.

\section{Multidimensional matching with a bilinear surplus}\label{sec: Matchingmodel}

We consider an environment where every worker with a bundle of skills sorts into a job demanding specific combinations of those skills. Let $\mathcal{X}\subset\mathbb{R}^{d}$ and $\mathcal{Y}\subset\mathbb{R}^{d}$ be spaces of worker and job characteristics endowed with probability measures $P$ and $Q$, respectively. Workers and jobs are described by the corresponding vectors of characteristics $x\in\mathcal{X}$ and $y\in\mathcal{Y}$. Every matched pair produces output according to production technology $s\left(x,y\right)$, and the surplus is shared through a competitive process that determines equilibrium wages $w^{*}(x)$ and profits $v^{*}(y)$. The equilibrium is characterized by the Walrasian stability condition: no worker-job pair can do better by deviating. As is well known in the optimal transport literature, under regularity conditions the stable matching can be formulated as a Monge-Kantorovich optimal transport problem whose solution determines a unique equilibrium assignment map $T^{*}:\mathcal{X}\to\mathcal{Y}$, mapping each worker $x$ to her equilibrium job $y=T^{*}(x)$, and whose dual yields the equilibrium wage and profit functions.\footnote{For the general theory, see \citet{villani2003, villani2008}, \cite{chiappori2010hedonic}, and \cite{galichon2017survey}. For reference on matching with transferable utility, see \cite{browning2014} and \cite{chiappori2017}. The bilinear surplus specification used in this paper dates back to at least \cite{tinbergen1956theory}, and has been employed in \citet{bojilov2016}, \citet{lindenlaub2017}, and the recent working paper by \cite{boerma2025}.}

We specify a bilinear surplus function, as in \cite{tinbergen1956theory}, \cite{bojilov2016}, and \cite{lindenlaub2017}:
\begin{equation}\label{eq: bilinear production function}
    s\left(x,y\right):=s\left(x,y;A,b\right)
    =x'Ay+x'b,
\end{equation}
where $A$ is a $d\times d$ matrix and $b$ is a $d\times1$ vector. The diagonal elements of $A$ capture within-task complementarities and the off-diagonal elements indicate between-task complementarities. $x'b$ represents non-interaction skill terms.\footnote{In the transferable utility framework with one-to-one matching, a linear term in $y$ can be absorbed into the firm-side profit function without affecting equilibrium matching or wages, so it is normalized to zero for identification and parsimony.} We keep the parametric surplus specification for two primary reasons: i) interpretability and comparability; and ii) identification power. The quadratic surplus provides a direct economic interpretation of the technology parameters (in $A$) as measures of production complementarity. As \cite{bojilov2016} advocated, the quadratic surplus is widely used in applied studies and provides simple, intuitive, and meaningful interpretations of worker-job interactions.\footnote{\cite{lindenlaub2025worker} also employ a similar bilinear surplus specification when they empirically implement their theoretical model with a general surplus, referring to the bilinear specification as a second-order Taylor series approximation.}

Matching assortativity depends crucially on the properties of the surplus function. To fix ideas, consider \citet{becker1973}'s spousal matching model where men and women are endowed with ``ability indices'' $x$ and $y$, respectively. If $\partial^{2}s\left(x,y\right)/\partial x\partial y\geq0$, then $T^{*}\left(x\right)=F_{y}^{-1}\left(F_{x}\left(x\right)\right)$ where $F_{x}$ and $F_{y}$ are the cumulative distribution functions of $x$ and $y$. $T^{*}$ having this property is defined as positive assortative matching (PAM) in the sense that high-type males match high-type females. Negative assortative matching (NAM) is the opposite. 

The properties of $A$ are pivotal to the assortativity of the equilibrium assignment in our specification. $T^{*}$ satisfies PAM if $A$ is diagonal with all positive principal minors (Proposition 2 of \citet{lindenlaub2017}). The assignment is unaffected by non-interaction terms because $\mathbb{E}_{\pi}\left[X'AY+X'b\right]=\mathbb{E}_{\pi}\left[X'AY\right]+\mathbb{E}_{P}\left[X'b\right]$ and the latter does not depend on $\pi$. Therefore, the dual optimal transport problem with $s\left(x,y\right)$ can be rewritten in terms of $s^{o}\left(x,y\right)=x'Ay$ as follows:
\begin{equation}\label{eq: dual-matching2}
\begin{split}
    &\inf_{w\in\mathcal{W}}\left\{\mathbb{E}_{P}\left[w\left(X\right)\right]
    +\mathbb{E}_{Q}\left[\sup_{x\in\mathcal{X}}
        \left\{s\left(x,Y\right)-w\left(x\right)\right\}\right]\right\}\\
    &=\inf_{w^{o}\in\mathcal{W}}\left\{\mathbb{E}_{P}\left[w^{o}\left(X\right)\right]
    +\mathbb{E}_{Q}\left[\sup_{x\in\mathcal{X}}
                            \left\{s^{o}\left(x,Y\right)-w^{o}\left(x\right)\right\}\right]\right\}
    +\mathbb{E}_{P}\left[X\right]'b.
\end{split}\end{equation}
Note that the solution $w^{*}$ to the problem with the original $s$ is obtained by $w^*\left(x\right)=w^{o*}\left(x\right)+x'b+c$ with any constant $c$ where $w^{o*}$ is the solution to the problem with $s^{o}$.

We now consider the case where the attributes of workers and firms are two-dimensional ($d=2$), tailoring the model to our data. Every worker is endowed with a bundle of cognitive and manual skills,
$x=\left(x_{C},x_{M}\right)$. In turn, each firm is endowed
with both cognitive and manual skill requirements, $y=\left(y_{C},y_{M}\right)$. With $A=\left(\left(\alpha_{CC},\alpha_{MC}\right)',\left(\alpha_{CM},\alpha_{MM}\right)'\right)$ and $b=\left(\beta_{C},\beta_{M}\right)'$, define $\delta:=\frac{\alpha_{MM}}{\alpha_{CC}}$, which represents the relative level of complementarities across cognitive and manual tasks. When both $\alpha_{CC}$ and $\alpha_{MM}$ are positive, $\delta < 1$ indicates that worker-job complementarity in the cognitive task is stronger than in the manual task.

Under joint normality of $x$ and $y$, \citet{lindenlaub2017} and \citet{bojilov2016} derive closed-form solutions for $T^{*}$ and $w^{*}$ and estimate the production technology parameters. In practice, however, $x$ and $y$ are non-normal as natural skills tend to have skewed distributions. To align the data with the model, \cite{lindenlaub2017} converts each element of $x$ and $y$ into standard normal. Their dependence is then modeled using a Gaussian copula. Figure \ref{fig: lindenlaub} illustrates how she derives the equilibrium assignment and wage function from transformed data. If the transformed data ($\tilde{x}$ and $\tilde{y}$) are bivariate normal, this transformation provides a way to study properties of matching in terms of technology parameters. However, as the equilibrium assignments are based on the original data, not the transformed data, $T^*(x)$ (optimal assignment map from the original data) and $\tilde{T}^*(x)$ (optimal assignment from transformed data) in Figure \ref{fig: lindenlaub} may not deliver the same empirical implications. Moreover, the joint distribution of $\tilde{x}$ and $\tilde{y}$ is not in general normal, and hence, the model can be misspecified even after transforming the data. To avoid such misspecification, we allow $x$ and $y$ to have any arbitrary distributions in our model.

\begin{figure}[tbh]
\begin{center}
\begin{tikzpicture}
\node at (0,0) (X) {$x\sim P$};
\node at (0,-3) (Y) {$y\sim Q$};
\node at (5,0) (X1) {$\begin{array}{l}\tilde{x}_{C}=\Phi^{-1}\left(F_{x_{C}}\left(x_{C}\right)\right)\sim{\rm N}\left(0,1\right)\\ \tilde{x}_{M}=\Phi^{-1}\left(F_{x_{M}}\left(x_{M}\right)\right)\sim{\rm N}\left(0,1\right)\end{array}$};
\node at (5,-3) (Y1) {$\begin{array}{l}\tilde{y}_{C}=\Phi^{-1}\left(F_{y_{C}}\left(y_{C}\right)\right)\sim{\rm N}\left(0,1\right)\\ \tilde{y}_{M}=\Phi^{-1}\left(F_{y_{M}}\left(y_{M}\right)\right)\sim{\rm N}\left(0,1\right)\end{array}$};
\node at (11.5,0) (X2) {$\begin{pmatrix}\tilde{x}_{C}\\ \tilde{x}_{M}\end{pmatrix}\sim{\rm N}\left(0,\begin{pmatrix}1 & \rho_{\tilde{x}}\\ \rho_{\tilde{x}} & 1\end{pmatrix}\right)$};
\node at (11.5,-3) (Y2) {$\begin{pmatrix}\tilde{y}_{C}\\ \tilde{y}_{M}\end{pmatrix}\sim{\rm N}\left(0,\begin{pmatrix}1 & \rho_{\tilde{y}}\\ \rho_{\tilde{y}} & 1\end{pmatrix}\right)$};
\draw [->] (X) -- (Y);
\draw [->] (X) -- (X1);
\draw [->] (X1) -- (X2);
\draw [->] (Y) -- (Y1);
\draw [->] (Y1) -- (Y2);
\draw [->] (X2) -- (Y2);
\node [above] at (8.5,0) {?};
\node [above] at (8.5,-3) {?};
\node [right] at (0,-1.5) {$T^{*}\left(x\right)$};
\node [left] at (11.5,-1.5) {$\tilde{T}^{*}\left(\tilde{x}\right)$};
\end{tikzpicture}
\end{center}
\caption{\label{fig: lindenlaub}\citet{lindenlaub2017}'s transformation. This figure should be read from left to right: starting from arbitrary marginals for $x$ and $y$, transformations are applied to obtain standard normal marginals, which are then associated by the Gaussian copula. The true dependence structure after transformations can differ from the Gaussian copula. $T^*(x)$ and $\tilde{T}^*(x)$ denote the optimal assignment maps in the original problem and in the transformed data, respectively. $\Phi$ denotes the standard normal c.d.f. $F_{x_{C}}$ and $F_{x_{M}}$ denote the c.d.f. for $x_{C}$ and $x_{M}$, respectively.}
\end{figure}

We further impose conditions on the probability measures $P$ and $Q$, as well as on $A$.
\begin{assumption}\label{assu:Px}
(i) $P$ and $Q$ have finite second moments, and (ii) $P$ is absolutely continuous with respect to the Lebesgue measure.
\end{assumption}
\begin{assumption}\label{assu:bilinearID}
The matrix $A$ in the production technology \eqref{eq: bilinear production function} is invertible.
\end{assumption}
\noindent The following proposition derives the equilibrium assignment and wage as the solution to the dual optimal transport problem \eqref{eq: dual-matching2}. This result builds on the characterization of optimal transport maps in \citet{villani2003, villani2008}. In the special case of jointly normally distributed characteristics, closed-form solutions are derived in \citet{lindenlaub2017} and (with entropy regularization) in \citet{bojilov2016}. Our contribution here is to state the result for general distributions under the bilinear surplus specification and to show that the equilibrium assignment is uniquely recovered from the gradient of the convex wage function.

\begin{prop}
\label{Prop:Matchingiden}Let Assumption \ref{assu:Px} hold. Then, there exists a unique (up to a constant) convex solution, $w^{o*}\left(x\right)$, to the second dual problem in \eqref{eq: dual-matching2}, and the equilibrium wage ($w^{*}$) and assignment ($y^{*}\left(x\right)=T^{*}\left(x\right)$ where $y^{*}\left(x\right)=(y_{1}^{*}\left(x\right),\ldots,y_{d}^{*}\left(x\right))'$) are given by
$$w^{*}\left(x\right)=w^{o*}\left(x\right)+x'b+c,\quad Ay^*\left(x\right)=\nabla w^{o*}\left(x\right),$$
where $c$ is the constant of integration. In addition, if Assumption \ref{assu:bilinearID} holds, the unique equilibrium assignment function is given by $y^*\left(x\right)=A^{-1}\nabla w^{o*}\left(x\right)$.
%and there exist three points $X_{1},X_{2},X_{3}$ such that $\nabla w_{0}\left(X_{1}\right)-\nabla w_{0}\left(X_{2}\right)$ and $\nabla w_{0}\left(X_{2}\right)-\nabla w_{0}\left(X_{3}\right)$ are linearly independent.
\end{prop}

\noindent We can interpret this problem as assigning from $\mathcal{X}$ to $A\mathcal{Y}:=\left\{Ay:y\in\mathcal{Y}\right\}$. Denoting $z =Ay$, the surplus function can be written as $s(x,z)=x'z$, which satisfies the twist condition, regardless of the invertibility of $A$. Hence, Assumption \ref{assu:Px} alone guarantees the existence of the convex solution, $w^{o*}\left(x\right)$, which implicitly depends on $A$. Assumption \ref{assu:bilinearID} is required to ensure the unique one-to-one mapping between $x$ and $y$. %In the dual problem \eqref{eq: dual-matching}, we may impose another constraint, $w\left(0\right)=0$, among the other constraints such as zero mean of $w\left(X\right)$. Then, $c$ equals wage of a worker with $X = \left(0,0\right)$.

\section{\label{sec: MatchingModelID}Econometric model and identification}

This section describes the econometric model using the theoretical results in the previous section. Let $\left\{\left(w_{i},x_{i}',y_{i}'\right)\right\}_{i=1}^{n}$ represent an i.i.d. sequence of $n$ matched observations on worker $i$'s wage $w_{i}$, her bundle of skills $x_{i}$, and the matched job's skill demands $y_{i}$. As Proposition \ref{Prop:Matchingiden} provides the equilibrium mapping from $x$ to $w$ and $y$ that is interpreted as Walrasian equilibrium, we can write down a system of equations for $w_i$ and $y_i$ in terms of $x_i.$ Pure Monge-Kantorovich optimal transport models (without unobserved heterogeneity) typically predict deterministic assignments between agent characteristics. To apply the model in empirical analysis, the models are typically regularized by introducing unobserved heterogeneity, search frictions, or measurement errors. Here we introduce measurement errors in the equilibrium wage and assignment functions to keep the econometric model in line with optimal transport theory. Then, we show that the equilibrium functions and technology parameters in $A$ are identified from the following nonparametric regression system.
%:
\begin{equation}\label{eq:OTmodel}\begin{split}
    w_{i}&=w^{*}\left(x_{i}\right)+\varepsilon_{wi}
    =w^{o*}\left(x_{i}\right)+x_{i}'b+c+\varepsilon_{wi},\\
    y_{i}&=y_{i}^{*}(x_i)+\varepsilon_{yi}
    =A^{-1}\nabla w^{o*}\left(x_{i}\right)+\varepsilon_{yi}.%\begin{pmatrix}\varepsilon_{Ci}\\ \varepsilon_{Mi}\end{pmatrix}.
\end{split}\end{equation}
Here, $\varepsilon_{wi}$ is a scalar measurement error in the observed wage. $\varepsilon_{yi}$ is a $d\times1$ vector of measurement errors in the firm's skill demands. 

In the empirical application, the introduction of measurement errors can be motivated by the construction of skill measures from data. \citet{sanders2014}, \citet{lindenlaub2017}, and \citet{lise2020} use the U.S. Department of Labor Occupational Characteristics Database (O*NET) to determine the levels of skills required for each categorical task. O*NET data provide rich information (more than 270 descriptors) on skill requirements for a large number of occupations. Measurement errors could arise in three ways. First, researchers conventionally classify the descriptors into predetermined skill categories, e.g., ``cognitive'' and ``manual''. However, this decision may be far from clear-cut for many descriptors. Second, the descriptors are aggregated within each category using principal component analysis. This procedure inevitably produces measurement errors even if the descriptors are correctly classified. Lastly, there may be unobserved factors not included in O*NET. As discussed in the Related Literature, unobserved heterogeneity also arises in other matching frameworks, but our additive and separable measurement errors are fundamentally different from the surplus-entering unobservables in the CS tradition.

The closed-form solutions in the Gaussian model involve the productivity correlation $\rho_{y^*} = \mathrm{corr}(y^*_C, y^*_M)$, which is the correlation of the latent equilibrium skill demands. However, the econometrician only observes $y = y^* + \varepsilon_y$, so the observed correlation $\mathrm{corr}(y_C, y_M)$ generally differs from $\rho_{y^*}$ due to attenuation from measurement error. In the Gaussian ML implementation of \citet{lindenlaub2017}, the observed correlation is used in place of the productivity correlation, which introduces a discrepancy in the likelihood function. Our sieve-based methods do not rely on the closed-form solution under bivariate normality, thereby avoiding this difficulty.

Unlike \cite{lindenlaub2017}'s approach, we do not need to impose distributional assumptions on measurement errors---assumptions that are vulnerable to misspecification. The moment conditions \eqref{eq:OTexoX} below correspond to what is usually referred to as \textit{weakly classical measurement error} in the terminology of \cite{schennach2016}. We impose the following conditions assuming the exogeneity of $x_i$:
\begin{equation}\label{eq:OTexoX}
    \mathbb{E}\left[\varepsilon_{wi}|x_{i}\right]=0,\quad
    \mathbb{E}\left[\varepsilon_{yi}|x_{i}\right]=0.
\end{equation}
Let $\theta=\left(\text{vec}\left(A^{-1}\right)',b'\right)'$
denote a vector of unknown finite-dimensional parameters and $\theta\in\Theta$
where $\Theta$ is a compact subset of $\mathbb{R}^{d^{2}+d}$. The normalizing constant is not in our parameters of interest, and henceforth, we
refer to $w\left(x\right):=w^{o*}\left(x\right)+c$ as the constant added infinite-dimensional parameter. We denote
$z_{i}=\left(w_{i},x_{i}',y_{i}'\right)'$
and $\rho\left(z_{i};\theta,w\right)=\left(\rho_{w}\left(w_{i},x_{i};\theta,w\right),\rho_{y}\left(y_{i},x_{i};\theta,w\right)'\right)'$,
where 
\begin{equation*}
    \rho_{w}\left(w_{i},x_{i};\theta,w\right)=w_{i}-\left(w\left(x_{i}\right)+x_{i}'b\right),\quad
    \rho_{y}\left(y_{i},x_{i};\theta,w\right)=y_{i}-A^{-1}\nabla w\left(x_{i}\right).
\end{equation*}

For each observation $i$, the model \eqref{eq:OTmodel} satisfies the moment conditions \eqref{eq:OTexoX}. This implies that the following conditional moments hold:
\begin{equation}\label{eq:CondMom}
    \mathbb{E}\left[\rho\left(z_{i};\theta,w\right)|x_{i}\right]=0,
\end{equation}
at the true parameter $\left(\theta_{0},w_{0}\right)$. Then $\left(\theta_{0},w_{0}\right)$ are identified via the model \eqref{eq:CondMom} by Proposition \ref{Prop:Matchingiden} and the exogeneity of $x_i$ as well as the following assumption on $\mathcal{Y}$, which requires that the support of job attributes is not contained in a lower-dimensional affine subspace.

\begin{assumption}\label{assu:bilinearID2}
There exist $y_{1},\ldots,y_{d},y_{d+1}\in\mathcal{Y}$ such that $\left\{y_{1}-y_{2},\ldots,y_{d}-y_{d+1}\right\}$ is linearly independent.
\end{assumption}

Proposition \ref{Prop:Matchingiden} implies the existence of a deterministic equilibrium characterized by a unique convex function $w_{0}$. When there is no non-interaction term $(b_{0} = 0)$, it follows that $\nabla w_{0}^{*} = \nabla w_{0}$. The strict convexity of $w_{0}$ further implies that $\mathbb{E}\left[\nabla w_{0}\left(x_{i}\right)\nabla w_{0}\left(x_{i}\right)^{\prime}\right]$ has a full rank, thus identifying $A_{0}$. Additionally, Assumption \ref{assu:bilinearID2} is sufficient to identify the nonzero vector $b_{0}$, as stated in the following theorem.

\begin{theorem}\label{thm:MatchingID}
Let Assumptions \ref{assu:Px}-\ref{assu:bilinearID2} hold and the moment conditions \eqref{eq:CondMom} be satisfied. Then, $\theta_{0}$ and $w_{0}=w_{0}^{o*}+c_{0}$ are identified.
\end{theorem}

We can identify $w_{0}^{o*}$ and $c_{0}$ separately under a normalization such as
$w_{0}^{o*}\left(x_{0}\right)=0$ for some $x_{0}\in\mathcal{X}$ or $\int_\mathcal{X}w_{0}^{o*}(x)dx=0$ 
when $\mathcal{X}$ is bounded. With the former constraint, $c_{0}$ and $w_{0}^{o*}\left(x\right)$ are 
identified with $w_{0}^{*}\left(x\right)=w_{0}^{o*}\left(x\right)+x'b$ since
$c_{0}=w_{0}^{o*}\left(x_{0}\right)+c_{0}=w_{0}^{*}\left(x_{0}\right)-x_{0}'b$.

\begin{remark}
    Our framework can be extended to allow for endogenous \(x_i\) using nonparametric instrumental-variables (NPIV) methods. Specifically, if one has instruments \(z_i\) satisfying \(E[\varepsilon_{wi}|z_i]=E[\varepsilon_{yi}|z_i]=0\), the equilibrium wage and assignment functions are identified from moment conditions of the form
  \[
    E\big[(w_i - w^\ast(x_i))\phi(z_i)\big]=0, \quad
    E\big[(y_i - y^\ast(x_i))\phi(z_i)\big]=0
  \]
   for a family of test functions \(\phi(\cdot)\), under standard completeness conditions. In this case, the technology parameters and the equilibrium wage function can be estimated by NPIV estimators, e.g., \cite{blundell2007}. NPIV estimation is in general an ill-posed problem. However, in our model, the equilibrium wage function is convex, so the finite sample properties can be improved by imposing a shape restriction as shown in \cite{chetverikov2017nonparametric}.
\end{remark}

\begin{remark}\label{remark:unequal}
Our framework extends to the case where the two sides of the market have different numbers of observable characteristics ($d_x\neq d_y$). Under the bilinear surplus, the matching problem is equivalent to optimally assigning $x$ to $Ay\in\mathbb{R}^{d_x}$, regardless of the dimension of $y$. This allows identification of the equilibrium wage function and the technology parameters without imposing $d_x=d_y$; see \cite{chiappori2017multi, chiappori2020multidimensional}, \cite{mccann2020optimal}, and \cite{nenna2023note} for general unequal-dimension optimal transport formulations, and Appendix~\ref{appen: unequal} for the formal identification result under our bilinear surplus specification. When $d_x \ge d_y$ and $A$ has rank $d_y$, the equilibrium assignment function $y^*(x) = (A'A)^{-1}A'\nabla w^{o*}(x)$ is also identified, and the matching between $x$ and $y$ is pure.
\end{remark}

\section{\label{sec:MatchingEst}Sieve-based semiparametric estimation}
The model parameters are identified by the semiparametric conditional moment restrictions \eqref{eq:CondMom}. If the function $w$ is parametrically specified, these moment conditions lead to standard GMM estimation. As $w$ is infinite-dimensional in our specification, we approximate it using sieves. The unknown function $w\in\mathcal{W}$ is approximated by $w_{n}\in\mathcal{W}_{n}$ where $\mathcal{W}_{n}$ is an approximating multivariate function space becoming dense in $\mathcal{W}$ as $n\rightarrow\infty$. We generate $\mathcal{W}_{n}$ via tensor-product construction. One could consider a kernel-based approximation for the equilibrium wage and assignment functions. However, using sieves (e.g., polynomials or splines) is particularly advantageous in our setting. First, the equilibrium assignment is determined by $y = \nabla w(x)$. Sieves allow us to compute this gradient analytically, so that the functional form information is easily incorporated in estimation. Second, the sieve objective functions are linear in parameters, so that obtaining the score and Hessian is relatively straightforward. Finally, our theory implies that $w$ is convex. The sieve framework accommodates global convexity constraints in a straightforward way, improving stability and computational efficiency.

From now on, we assume that there are sets of firms and workers with $d=2$. Every worker is endowed with a bundle of cognitive and manual skills, $x=\left(x_{C},x_{M}\right)$. In turn, each firm is endowed with both cognitive and manual skill demands, $y=\left(y_{C},y_{M}\right)$. $y_{C}$ ($y_{M}$) corresponds to the productivity or skill requirement of cognitive task $C$ (manual task $M$). In our case, 
{\small\begin{equation}\label{eq: funcspace}
    \mathcal{W}_{n}
    =\left\{w_{n}:\mathcal{X}\to\mathbb{R},
        w_{n}\left(x;\gamma\right)=\sum_{j_{C}=0}^{k_{Cn}}\sum_{j_{M}=0}^{k_{Mn}}
            \gamma_{j_{C}j_{M}}p_{j_{C}}\left(x_{C}\right)p_{j_{M}}\left(x_{M}\right),
            \gamma_{j_{C}j_{M}}\in\mathbb{R}\right\},
\end{equation}}%
where $\left\{p_{j_{C}}\left(x_{C}\right)\right\}_{j_{C}=0}^{k_{Cn}}$ and $\left\{p_{j_{M}}\left(x_{M}\right)\right\}_{j_{M}=0}^{k_{Mn}}$ are basis functions of $x_{C}$ and $x_{M}$. The tensor-product space is simple to extend to higher dimensions and easy to implement. For our second and third conditional moment restrictions, we approximate $\partial w_{0}\left(x\right)/\partial x_{C}$ and $\partial w_{0}\left(x\right)/\partial x_{M}$ with the same parameter values $\left\{\gamma_{j_{C}j_{M}}\right\}$ used to approximate $w_{0}\left(x\right)$ in $\mathcal{W}_{n}$:
\begin{equation*}\begin{split}
    \partial w_{n}\left(x;\gamma\right)/\partial x_{C}
            &=\sum_{j_{C}=0}^{k_{Cn}}\sum_{j_{M}=0}^{k_{Mn}}
            \gamma_{j_{C}j_{M}}\left(\partial p_{j_{C}}\left(x_{C}\right)/\partial x_{C}\right)
            p_{j_{M}}\left(x_{M}\right),\\
    \partial w_{n}\left(x;\gamma\right)/\partial x_{M}
            &=\sum_{j_{C}=0}^{k_{Cn}}\sum_{j_{M}=0}^{k_{Mn}}
            \gamma_{j_{C}j_{M}}p_{j_{C}}\left(x_{C}\right)
            \left(\partial p_{j_{M}}\left(x_{M}\right)/\partial x_{M}\right).
\end{split}\end{equation*}

We first consider the model \eqref{eq:OTmodel} with normally distributed, mean-zero measurement errors that are mutually uncorrelated. Then, we can estimate the parameters using sieve maximum likelihood (SML). Assuming $\varepsilon_{i}\sim N\left(0,\Sigma\right)$, we
write the log-likelihood function of model \eqref{eq:OTmodel} as
\[
    L^{*}\left(\theta,\Sigma,w\left(\cdot\right)\right)
    =-\frac{n}{2}\log\det\left(\Sigma\right)
    +\sum_{i=1}^{n}\log\left|\det\left(\partial\rho_{i}
                             /\partial\left(w_{i},y_{Ci},y_{Mi}\right)\right)\right|
    -\frac{1}{2}\sum_{i=1}^{n}\rho_{i}'\Sigma^{-1}\rho_{i},
\]
where $\rho_{i}=\rho\left(z_{i};\theta,w\right)$. Solving $\partial L^{*}/\partial\Sigma=0$ for $\Sigma$, we get $\Sigma=\frac{1}{n}\sum_{i=1}^{n}\rho_{i}\rho_{i}'$, which yields the concentrated log-likelihood function of our model
\begin{equation}\label{eq: SieveNLFI}
    L\left(\theta,w\left(\cdot\right)\right)
    =-\frac{n}{2}\log\det\left(\frac{1}{n}\sum_{i=1}^{n}\rho_{i}\rho_{i}'\right).
\end{equation}
The value of $\left(\theta,w\right)$ maximizing \eqref{eq: SieveNLFI} is the sieve nonlinear full information maximum likelihood estimator of $\left(\theta,w\right)$.

%\begin{theorem}\label{thm:AsympSieveMLE}
%Suppose Assumptions 
%\end{theorem}
%\begin{proof}
%    See Appendix X.
%\end{proof}

The normality assumption on measurement errors has no theoretical or empirical basis. Without any distributional assumptions on measurement errors, we can still estimate $\left(\theta,w\right)$ using several sieve M-estimators. As $\rho\left(z;\theta,w\right)-\rho\left(z;\theta_{0},w_{0}\right)$ does not depend on $y$ under Assumption \ref{assu:bilinearID}, we can apply the sieve generalized least squares (GLS) procedure \citep{chen2007} that minimizes the following objective function with respect to $(\theta, w)$:
\[
\min_{\left(\theta,w\right)}\sum_{i=1}^{n}\rho\left(z_{i};\theta,w\right)^{\prime}\left[\hat{\Sigma}_{0}\left(x_{i}\right)\right]^{-1}\rho\left(z_{i};\theta,w\right),
\]
where $\hat{\Sigma}_{0}\left(x\right)$ is a consistent estimator
of the optimal weighting matrix $\Sigma_{0}\left(x\right):={\rm Var}\left[\left.\rho\left(z_{i};\theta,w\right)\right|x_{i}=x\right].$
In addition, if $A$ is diagonal, we can rewrite the last two moment conditions as 
$\rho_{C}\left(y_{C},x;\kappa_{C},w\right)=y_{C}-\kappa_{C}\nabla_{C}w\left(x\right)$ and $\rho_{M}\left(y_{M},x;\kappa_{M},w\right)=y_{M}-\kappa_{M}\nabla_{M}w\left(x\right)$, where $\kappa_{C}=\alpha_{CC}^{-1}$ and $\kappa_{M}=\alpha_{MM}^{-1}$.
Table \ref{tab: SGLS algorithm} outlines the three-step procedure to compute the SGLS estimator. 

\begin{table}[ht!]
\caption{Three-step procedure for Sieve GLS estimation \citep{chen2007}}\label{tab: SGLS algorithm}
\noindent \centering{}%
\begin{tabular}{l}
\hline 
\textbf{\textsc{Algorithm:}} Computing the Sieve GLS Estimator
of $\theta$ and $w$\tabularnewline
\hline 
1. Obtain an initial consistent sieve LS estimator $\left(\tilde{\theta}_{n},\tilde{w}_{n}\right)$
by\tabularnewline
$\quad \ \min_{\left(\theta,w\right)}\sum_{i=1}^{n}\rho\left(z_{i};\theta,w\right)^{\prime}\rho\left(z_{i};\theta,w\right),$\tabularnewline
2. Obtain a consistent estimator $\hat{\Sigma}_{0}\left(x\right)$ of
$\Sigma_{0}\left(x\right)={\rm Var}\left[\left.\rho\left(z_{i};\theta,w\right)\right|x_{i}=x\right]$\tabularnewline
$\quad \ $using $\left(\tilde{\theta}_{n},\tilde{w}_{n}\right)$ and
sieve LS estimation.\tabularnewline
3. Obtain the optimally weighted sieve GLS estimator $\left(\hat{\theta}_{n},\hat{w}_{n}\right)$
by\tabularnewline
$\quad \ \min_{\left(\theta,w\right)}\sum_{i=1}^{n}\rho\left(z_{i};\theta,w\right)^{\prime}\left[\hat{\Sigma}_{0}\left(x_{i}\right)\right]^{-1}\rho\left(z_{i};\theta,w\right)$.\tabularnewline
\hline 
\end{tabular}
\end{table}

The SGLS estimator allows for arbitrary correlation between measurement errors and heteroskedasticity. We can impose homoskedasticity by assuming $\Sigma_{0}\left(x\right)=\Sigma_{0}$ so that the optimal weighting matrix does not vary with $x$. If we further assume that the measurement errors are uncorrelated, i.e., $\Sigma_0$ is diagonal, we can use the sieve least squares (SLS) estimator from step 1 of the three-step procedure. We summarize the key differences in assumptions imposed in different estimation procedures in Table \ref{tab:comparison-estimators}.

\begin{table}[ht!]
\caption{Comparison of key assumptions in estimation procedures}
\label{tab:comparison-estimators}
\centering
\begin{tabular}{lcccc}
\hline
\multirow{2}{*}{Assumptions} & \cite{lindenlaub2017} & \multicolumn{3}{c}{Sieve Estimators} \\
\cmidrule(lr){2-2} \cmidrule(lr){3-5}
 & ML & SML & SLS & SGLS \\
\hline
Joint normality of $x$ and $y$ & {\checkmark} & {} & {} & {} \\
Normality of measurement errors & {\checkmark} & {\checkmark} & {} & {} \\
Uncorrelated measurement errors & {\checkmark} & {\checkmark} & {} & {} \\
Homoskedasticity of measurement errors & {\checkmark} & {\checkmark} & {\checkmark} & {} \\
\hline
\end{tabular}
\end{table}

We implement the sieve estimators using finite-dimensional Bernstein polynomials to construct the approximating space $\mathcal{W}_{n}$ of $\mathcal{W}$ on $\left[0,1\right]^{2}$. The basis functions are
$p_{j_{C}}\left(x_{C}\right)
    =\binom{k_{C_{n}}}{j_{C}}\left(x_{C}\right)^{j_{C}}\left(1-x_{C}\right)^{k_{C_{n}}-j_{C}}$ and 
$p_{j_{M}}\left(x_{M}\right)
    =\binom{k_{M_{n}}}{j_{M}}\left(x_{M}\right)^{j_{M}}\left(1-x_{M}\right)^{k_{M_{n}}-j_{M}}$, where
$j_{C}=0,1,\ldots,k_{C_{n}}$, $j_{M}=0,1,\ldots,k_{M_{n}}$, and $\binom{k}{j}$ is a binomial coefficient.\footnote{$x$ does not lie in $\left[0,1\right]^{2}$ in many applications. To satisfy the domain restriction for our simulation studies and empirical application, we use the following linear transformation when $\left(x_{C},x_{M}\right)\in\left[\underline{x}_{C},\overline{x}_{C}\right]\times\left[\underline{x}_{M},\overline{x}_{M}\right]$: $p_{j_{C}}\left(x_{C}\right)
    =\binom{k_{C_{n}}}{j_{C}}x_{1}^{j_{C}}\left(1-x_{1}\right)^{k_{C_{n}}-j_{C}}$ and 
$p_{j_{M}}\left(x_{M}\right)
    =\binom{k_{M_{n}}}{j_{M}}x_{2}^{j_{M}}\left(1-x_{2}\right)^{k_{M_{n}}-j_{M}}$,
where $x_{1}=(x_{C}-\underline{x}_{C})/(\overline{x}_{C}-\underline{x}_{C})$ and $x_{2}=(x_{M}-\underline{x}_{M})/(\overline{x}_{M}-\underline{x}_{M})$.}

If $\gamma_{j_{C}j_{M}}=w\left(j_{C}/k_{C_{n}},j_{M}/k_{M_{n}}\right)$, the Bernstein polynomial $w_{n}\left(x;\gamma\right)$ converges uniformly to $w(x)$ by the Stone-Weierstrass approximation theorem (see, e.g., \citet{lorentz1986}). %Furthermore, if $w\left(X\right)$ is twice continuously differentiable, then the order of the approximation error is $O\left(1/\min\left\{k_{C_{n}},k_{M_{n}}\right\}\right)$ (see, e.g., \citet{powell1981}).
This provides an approach to imposing shape restrictions on the sieve estimator with a linear constraint which can be solved easily.\footnote{\citet{compiani2022market} uses linear constraints for the function $w$ to impose monotonicity restrictions and a so-called ``diagonal dominance'' constraint.} Without any constraint, the equilibrium wage function, $w\left(x\right)+x'b$, is unique and convex. To obtain a more stable estimator, without loss of generality, we impose linear constraints on the Bernstein polynomials, which are necessary for the function to be convex. A detailed description of implementing this convexity constraint in the estimation procedures is provided in Appendix \ref{appen: bernstein convexity}.

To understand technological changes in the production function, the parametric components of the model are of primary interest. The SGLS estimator is ideal in this case because for $\theta$ (i) it is easy to use, $\sqrt{n}$-consistent, and asymptotically normal; (ii) it is semiparametrically efficient; and (iii) the asymptotic variance estimator of $\hat{\theta}$ is consistent and easy-to-compute.\footnote{The sieve minimum distance estimator can be considered. However, when $\rho\left(z;\theta,w\right)-\rho\left(z;\theta_{0},w_{0}\right)$ does not depend on $y$, the SGLS estimator is simpler to implement and computationally faster.} We formally derive its asymptotic properties in the following section. 

\section{Asymptotic theory for the SGLS estimator}\label{sec: asymptotics}
We establish consistency, convergence rate, asymptotic normality, and semiparametric efficiency of our SGLS estimator using results in \cite{chen1998}, \cite{ai2003}, and \cite{chen2007}. Define $\lambda := (\theta, w(\cdot)).$ Let $\hat{\lambda}_n$ and $\lambda_0$ denote our sieve GLS estimator and the true parameter values, respectively. We first show that $\hat{\lambda}_{n}$ converges to $\lambda_{0}$ at a rate faster than $n^{-1/4}$ under a pseudo norm $\lVert\cdot\rVert$. For any $\lambda_{1}=\left(\theta_{1},w_{1}\left(\cdot\right)\right),\lambda_{2}=\left(\theta_{2},w_{2}\left(\cdot\right)\right)\in\Lambda$, $\lVert\cdot\rVert$ is defined as
\[
    \lVert\lambda_{1}-\lambda_{2}\rVert^{2}
    =\mathbb{E}\left[\left(\frac{d\rho\left(z_{i};\lambda_{0}\right)}{d\lambda}
                     \left[\lambda_{1}-\lambda_{2}\right]\right)'
                     \Sigma\left(x_{i}\right)^{-1}
                     \left(\frac{d\rho\left(z_{i};\lambda_{0}\right)}{d\lambda}
                     \left[\lambda_{1}-\lambda_{2}\right]\right)\right],
\]
where
\[
    \frac{d\rho\left(z;\lambda_{0}\right)}{d\lambda}
    \left[\lambda_{1}-\lambda_{2}\right]
    =\begin{pmatrix}
        w_{1}\left(x\right)-w_{2}\left(x\right)+x'\left(b_{1}-b_{2}\right) \\
        \nabla_{C}w_{0}\left(x\right)\left(\kappa_{C1}-\kappa_{C2}\right)
        +\kappa_{C0}\left(\nabla_{C}w_{1}\left(x\right)-\nabla_{C}w_{2}\left(x\right)\right)\\
        \nabla_{M}w_{0}\left(x\right)\left(\kappa_{M1}-\kappa_{M2}\right)
        +\kappa_{M0}\left(\nabla_{M}w_{1}\left(x\right)-\nabla_{M}w_{2}\left(x\right)\right)
    \end{pmatrix}.
\]
The pseudo metric is comparatively weaker than the standard sup or $L_{2}$ metric, wherein convergence of $\hat{\lambda}_n$ to $\lambda_0$ under the standard metric implies convergence under the pseudo metric. \cite{ai2003} show that $\hat{\lambda}_n$ converging at a rate faster than $n^{-1/4}$ under the weaker metric $||\cdot||$ suffices to derive the $\sqrt{n}$-asymptotic normality of the parametric component, $\hat{\theta}_n$. 
%Euclidean metric for $\hat{\theta}_{n}$ and $L_{2}$ metric for $\hat{w}_{n}$, respectively.

Let $\Lambda=\Theta\times\mathcal{W}$ be equipped with a norm $\lVert\lambda\rVert_{s}=\left|\theta\right|_{e}+\lVert w\rVert_{\infty}
+\lVert\nabla_{C}w\rVert_{\infty}+\lVert\nabla_{M}w\rVert_{\infty}$,
where $\left|\cdot\right|_{e}$ denotes the Euclidean norm and $\lVert w\rVert_{\infty}=\sup_{x\in\mathcal{X}}\left|w\left(x\right)\right|$ is the supremum norm. We introduce the H\"{o}lder class of functions. Let $\left[m\right]$ be the largest nonnegative integer such that $\left[m\right]<m$. A real-valued function $w$ on $\mathcal{X}$ is said to be in H\"{o}lder space $\Lambda^{m}\left(\mathcal{X}\right)$  if it is $\left[m\right]$ times continuously differentiable on $\mathcal{X}$ and
\[
    \max_{\ell_{1}+\ell_{2}\leq\left[m\right]}\sup_{x}
        \left|\frac{\partial^{\ell_{1}+\ell_{2}}w\left(x\right)}
                   {\partial x_{C}^{\ell_{1}}\partial x_{M}^{\ell_{2}}}\right|+
    \sup_{m_{1}+m_{2}=\left[m\right]}\sup_{x,x'}
    \left|\frac{\partial^{\left[m\right]}w\left(x\right)}{\partial x_{C}^{m_{1}}\partial x_{M}^{m_{2}}}
    -\frac{\partial^{\left[m\right]}w\left(x'\right)}{\partial x_{C}^{m_{1}}\partial x_{M}^{m_{2}}}\right|
    /\left|x-x'\right|_{e}^{m-\left[m\right]}
\]
is finite. We provide the following assumptions for convergence.

\begin{assumption}\label{assu:IID}
(i) $\left\{w_{i},y_{i}',x_{i}'\right\}_{i=1}^{n}$ are i.i.d.; (ii) $\mathcal{X}$ is compact and a Cartesian product of compact intervals $\mathcal{X}_{C}$ and $\mathcal{X}_{M}$.
\end{assumption}
\begin{assumption}\label{assu:Sigma}
$\Sigma\left(x\right)$ and $\Sigma_{0}\left(x\right)\equiv\mathrm{Var}\left(\rho\left(z_{i},\lambda_{0}\right)|x_{i}=x\right)$ are positive definite and bounded uniformly over $x\in\mathcal{X}$.
\end{assumption}
\begin{assumption}\label{assu:compact}
$\Lambda\equiv\Theta\times\mathcal{W}$ is compact under $\lVert\cdot\rVert_{s}$.
\end{assumption}
\begin{assumption}\label{assu:Holder}
(i) $w\in\Lambda^{m}\left(\mathcal{X}\right)$ with $m>2$; (ii) $\forall w\in\Lambda^{m}\left(\mathcal{X}\right), \exists w_{n}\left(x;\gamma\right)\in\mathcal{W}_{n}$ such that $\lVert w_{n}-w\rVert_{\infty}=O(\left(k_{Cn}k_{Mn}\right)^{-m/2})$ with $k_{Cn},k_{Mn}=O\left(n^{1/2\left(m+1\right)}\right)$.
\end{assumption}

Assumption \ref{assu:IID} requires the data to be i.i.d. and the support of $x$ to be compact. Assumption \ref{assu:Sigma} ensures the existence and invertibility of the weight matrix for GLS estimation. Assumption \ref{assu:compact} requires the parameter space to be compact under the sup norm. We do not explicitly require identification of $\lambda$ here as Assumptions \ref{assu:Px}--\ref{assu:bilinearID2} guarantee it by Theorem \ref{thm:MatchingID}. Assumption \ref{assu:Holder} quantifies the sieve approximation error. Most papers in the literature require $m>d_{\mathcal{X}}/2$, where $d_{\mathcal{X}}$ is the dimension of $\mathcal{X}$. However, our objective function involves $\nabla_{C}w\left(x\right)$ and $\nabla_{M}w\left(x\right)$, so we need a higher order of $m$. Note that the smoothness of $w$ can be verified using the theory of optimal transport. The degree of the smoothness of the solution function $w$ depends on how smooth the densities of $x$ and $y^{*}$ are.\footnote{Assuming the densities are bounded away from zero and infinity, if the densities of variables $x$ and $y^{*}$ belong to the space $\Lambda^{m-2}$, the function $w_{0}$ is a member of $\Lambda^{m}\left(\mathcal{X}\right)$. For a more comprehensive understanding, refer to \citet{caffarelli1992CPAM,caffarelli1992JAMS,caffarelli1996} which cover the case of compactly supported $\mathcal{X}$ and $\mathcal{Y}^{*}.$ \citet{cordero-erausquin2019} provides an extended result for distributions with unbounded supports.} The following proposition establishes the convergence rate of $\hat{\lambda}_n.$

\begin{prop}\label{prop:convrate}
    If Assumptions \ref{assu:Px}-\ref{assu:Holder} hold, then $\lVert\hat{\lambda}_{n}-\lambda_{0}\rVert=o_{p}\left(n^{-1/4}\right)$.
\end{prop}

We now derive the asymptotic normality of the parametric components of the SGLS estimator, $\hat{\theta}_{n}$. Define
$D_{v}\left(x\right):=\left(D_{v_{1}}\left(x\right),D_{v_{2}}\left(x\right),D_{v_{3}}\left(x\right),D_{v_{4}}\left(x\right)\right)$ where
\begin{align*}
    D_{v_{1}}\left(x\right)
    &=\begin{pmatrix}v_{1}\left(x\right)\\
    \kappa_{C0}\nabla_{C}v_{1}\left(x\right)-\nabla_{C}w_{0}\left(x\right)\\
    \kappa_{M0}\nabla_{M}v_{1}\left(x\right)\end{pmatrix},\
    &D_{v_{2}}\left(x\right)
    &=\begin{pmatrix}v_{2}\left(x\right)\\
    \kappa_{C0}\nabla_{C}v_{2}\left(x\right)\\
    \kappa_{M0}\nabla_{M}v_{2}\left(x\right)-\nabla_{M}w_{0}\left(x\right)\end{pmatrix},\\
    D_{v_{3}}\left(x\right)
    &=\begin{pmatrix}v_{3}\left(x\right)-x_{C}\\
    \kappa_{C0}\nabla_{C}v_{3}\left(x\right)\\
    \kappa_{M0}\nabla_{M}v_{3}\left(x\right)\end{pmatrix},\
    &D_{v_{4}}\left(x\right)
    &=\begin{pmatrix}v_{4}\left(x\right)-x_{M}\\
    \kappa_{C0}\nabla_{C}v_{4}\left(x\right)\\
    \kappa_{M0}\nabla_{M}v_{4}\left(x\right)\end{pmatrix}.
\end{align*}
Let $v^{*}=\left(v_{1}^{*},v_{2}^{*},v_{3}^{*},v_{4}^{*}\right)$, where $v_{j}^{*}$ solves
\begin{equation}\label{eq: v*}
    \inf_{v_{j}}\mathbb{E}\left[D_{v_{j}}\left(x_{i}\right)'
    \Sigma\left(x_{i}\right)^{-1}D_{v_{j}}\left(x_{i}\right)\right].
\end{equation}
\begin{assumption}\label{assu:Dv*}
(i) $\mathbb{E}\left[D_{v^{*}}\left(x_{i}\right)'D_{v^{*}}\left(x_{i}\right)\right]$ is positive definite; (ii) Each element of $v^{*}$ belongs to the H\"{o}lder space $\Lambda^{m}\left(\mathcal{X}\right)$ with $m>2$.
\end{assumption}
\begin{assumption}\label{assu:interior}
$\theta_{0}\in\mathrm{int}\left(\Theta\right)$.
\end{assumption}
Assumptions \ref{assu:Dv*}--\ref{assu:interior} are standard conditions for M-estimation that guarantee the uniqueness of the solution within the parameter space. Under Assumptions \ref{assu:Px}--\ref{assu:Dv*}, it follows from Lemma B.1 in \citet{ai2003} that $|\hat{\theta}_{n}-\theta_{0}|_{e}=o_{p}\left(n^{-1/4}\right)$, $\lVert\hat{w}_{n}-w_{0}\rVert_{2}=\left(\mathbb{E}\left[\left(\hat{w}_{n}\left(x_{i}\right)-w_{0}\left(x_{i}\right)\right)^{2}\right]\right)^{1/2}=o_{p}\left(n^{-1/3}\right)$, and $\lVert\nabla_{C}\hat{w}_{n}-\nabla_{C}w_{0}\rVert_{2},\lVert\nabla_{M}\hat{w}_{n}-\nabla_{M}w_{0}\rVert_{2}=o_{p}\left(n^{-1/4}\right)$. Now the following theorem provides the asymptotic normality of $\hat{\theta}_n.$

\begin{theorem}\label{thm:AsympSieveGLSE}
Let Assumptions \ref{assu:Px}--\ref{assu:interior} hold. Then, $\sqrt{n}(\hat{\theta}_{n}-\theta_{0})\rightarrow_{d}N\left(0,V_{1}^{-1}V_{2}V_{1}^{-1}\right)$, where
\[\begin{split}
    V_{1}
    &=\mathbb{E}\left[D_{v^{*}}\left(x_{i}\right)'\Sigma\left(x_{i}\right)^{-1}D_{v^{*}}\left(x_{i}\right)\right],\\
    V_{2}
    &=\mathbb{E}\left[D_{v^{*}}\left(x_{i}\right)'
                     \Sigma\left(x_{i}\right)^{-1}\Sigma_{0}\left(x_{i}\right)\Sigma\left(x_{i}\right)^{-1}
                     D_{v^{*}}\left(x_{i}\right)\right].
\end{split}\]
\end{theorem}

The asymptotic variance $V_{1}^{-1}V_{2}V_{1}^{-1}$ can be consistently estimated (see Remark 4.2 in \citet{chen2007}) and the standard errors of $\left(\hat{\alpha}_{CC},\hat{\alpha}_{MM}\right)=\left(1/\hat{\kappa}_{C},1/\hat{\kappa}_{M}\right)$ are obtained by using the delta method. Furthermore, if all conditions of Theorem \ref{thm:AsympSieveGLSE} are satisfied with $\Sigma\left(x\right)=\Sigma_{0}\left(x\right)$, $\hat{\theta}_{n}$ achieves semiparametric efficiency with a consistent estimator $\hat{\Sigma}_{0}\left(x\right)$ of $\Sigma_{0}\left(x\right)$. The estimation of $\Sigma_{0}\left(x\right)$ is straightforward through series least squares estimation, using the initial consistent SLS estimator $(\tilde{\theta}_{n},\tilde{w}_{n})$. To ensure the efficiency of the SGLS estimator, $\hat{\Sigma}_{0}\left(x\right)$ is required to exhibit the following uniform convergence rate.

\begin{assumption}\label{assu:Sigma2}
$\hat{\Sigma}_{0}\left(x\right)=\Sigma_{0}\left(x\right)+o_{p}\left(n^{-1/4}\right)$ uniformly over $x\in\mathcal{X}$.
\end{assumption}

Let $v_{0}=\left(v_{01},v_{02},v_{03},v_{04}\right)$, where $v_{0j}$ solves \eqref{eq: v*} with $\Sigma\left(x\right)$ replaced by $\Sigma_{0}\left(x\right)$. Now the following theorem establishes the semiparametric efficiency of $\hat{\theta}_n$.
\begin{theorem}\label{thm:EfficSieveGLSE}
Suppose that all conditions of Theorem \ref{thm:AsympSieveGLSE} with $\Sigma\left(x\right)=\Sigma_{0}\left(x\right)$ and $v^{*}=v_{0}$ hold, and Assumption \ref{assu:Sigma2} is satisfied. Then, $\sqrt{n}(\hat{\theta}_{n}-\theta_{0})\rightarrow_{d}N\left(0,V_{0}^{-1}\right)$, with $V_{0}=\mathbb{E}\left[D_{v_{0}}\left(x_{i}\right)'\Sigma_{0}\left(x_{i}\right)^{-1}D_{v_{0}}\left(x_{i}\right)\right]$.
\end{theorem}

Note that the assumptions introduced in this section are standard regularity conditions for sieve M-estimation. They do not impose additional restrictive structures compared to the assumptions required for \cite{lindenlaub2017}'s maximum likelihood (ML) estimation. Rather, they are the nonparametric generalizations of the standard regularity conditions for ML estimators.\footnote{For example, Assumptions \ref{assu:IID}, \ref{assu:Sigma}, \ref{assu:Dv*}, and \ref{assu:interior} have exact counterparts in standard finite-dimensional MLE theory. Assumptions \ref{assu:compact}, \ref{assu:Holder}, and \ref{assu:Sigma2} are merely the technical conditions required to allow the underlying distributions of $x$ and $y$ to be ``unknown smooth distributions''.}

\section{\label{sec: Simul}Monte Carlo simulations}
This section evaluates the finite sample performances of our sieve estimators using known data-generating processes (DGPs). We first generate Monte Carlo samples from \cite{lindenlaub2017}'s quadratic-Gaussian model. Workers' skill bundle, $x$, and occupations' skill requirements, $y$, follow joint Gaussian distributions:
$$  \begin{pmatrix} x_C\\x_M \end{pmatrix}
    \sim N\left(\begin{pmatrix} 0\\0 \end{pmatrix},
                \begin{pmatrix} 1&\rho_x \\ \rho_x&1 \end{pmatrix}\right),\quad
    \begin{pmatrix} y_C\\y_M \end{pmatrix}
    \sim N\left(\begin{pmatrix} 0\\0 \end{pmatrix},
                \begin{pmatrix} 1&\rho_y \\ \rho_y&1 \end{pmatrix}\right),$$
from which $\left\{x_i\right\}_{i=1}^n$ are drawn with sample size $n=3000$. Following \cite{lindenlaub2017}, we rule out between-task complementarities in the production technology \eqref{eq: bilinear production function} so that $A$ is a diagonal matrix $(A={\rm diag}\left(\alpha_{CC},\alpha_{MM}\right))$. Then, the equilibrium assignment $y^*$ and wage $w^*$ have closed-form solutions as shown in \cite{lindenlaub2017}. Both depend on $\rho_x, \rho_y,$ and technology parameters. We generate $\left\{y_i^*\right\}_{i=1}^n$ and the equilibrium wages $\left\{w_i^*\right\}_{i=1}^n$ using closed-form solutions. Lastly, we draw measurement errors from mean-zero Gaussian distributions:
$$  \varepsilon_w \sim N(0,\sigma_w^2),\quad \varepsilon_C \sim N(0,\sigma_C^2),\quad
    \varepsilon_M \sim N(0,\sigma_M^2),$$
and add them to $w_i^*$, $y_{Ci}^*$, and $y_{Mi}^*$, respectively, to generate the observable data $(w_i, y_i, x_i)_{i=1}^n$ following \eqref{eq:OTmodel}. The true parameter values used in simulations are:
$$  (\alpha_{CC},\alpha_{MM},\beta_C,\beta_M,c,\rho_x,\rho_y,\sigma_w,\sigma_C,\sigma_M)
    =(0.5,0.2,1.7,-0.4,30,-0.4,-0.5,2,1,1),$$
which are set close to the empirical ML estimates in \cite{lindenlaub2017}. 

We estimate the production technology parameters using \cite{lindenlaub2017}'s parametric ML estimator and our sieve estimators (SML, SLS, and SGLS) across 1000 Monte Carlo samples. As we discussed earlier, the original ML estimator suffers from bias because the solution uses the measurement-error-contaminated productivity correlation $\tilde{\rho}_y = corr(y)$, which is different from the true productivity correlation $\rho_y = corr(y^*).$ If the measurement errors in $y$ are negligible, e.g., $(\sigma_C, \sigma_M)$ are close to 0, the ML estimator works well for this DGP. However, given the current parameter specification, the measurement errors are substantial, so the ML estimator can be inconsistent. To address this issue, we define a corrected ML estimator (referred to as `ML$^*$') that uses the corrected productivity correlation $\rho_y=\tilde{\rho}_y \sqrt{var(y_1)var(y_2)}/\sqrt{(var(y_1)-\sigma_C^2)(var(y_2)-\sigma_M^2)}.$ This correction in turn yields much more precise estimates than the original ML estimator. The sieve estimators do not share this problem. 

\begin{figure}[ht!]
\begin{centering}
\includegraphics[width=\textwidth]{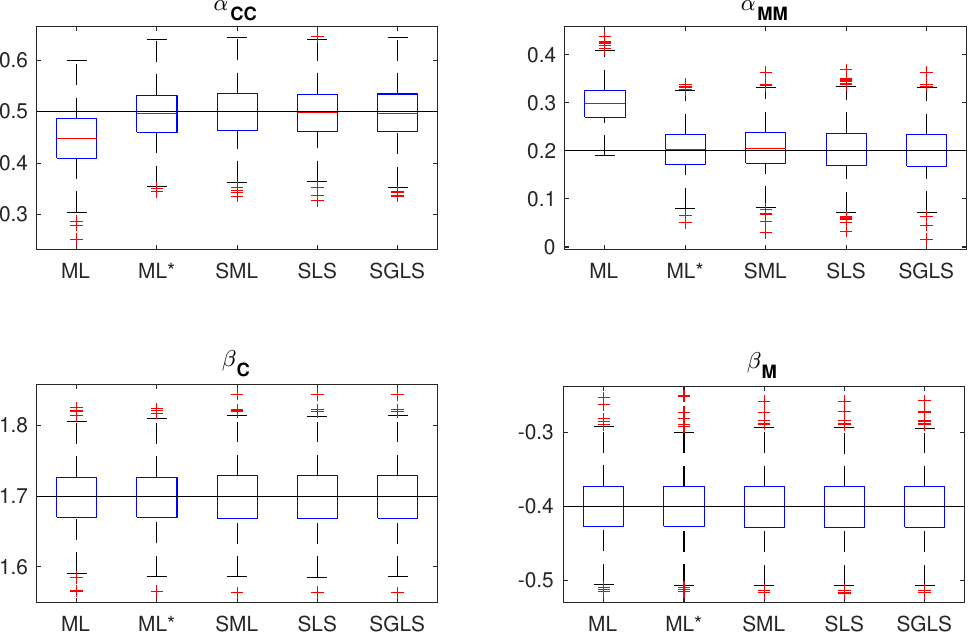}
\par\end{centering}
\caption{\label{fig:simul performance}Box plots of parameter estimates (Gaussian DGP)}
\end{figure}

The box plots in Figure \ref{fig:simul performance} summarize the distributions of parameter estimates delivered by the 5 estimators we consider. For $(\beta_C, \beta_M),$ all the estimators work well. On the other hand, for $(\alpha_{CC}, \alpha_{MM})$, the original ``ML'' estimator shows substantial bias due to the use of the observed (rather than latent) productivity correlation in the likelihood. Our corrected ``ML$^*$'' estimator, on the contrary, accurately recovers the true parameters and performs the best under the Gaussian DGP as it relies on the true specification. In this regard, we focus on ``ML$^*$'' as the appropriate parametric benchmark and omit the ``ML'' results in subsequent analyses. The sieve estimators also perform very well for $(\alpha_{CC}, \alpha_{MM})$ as the distributions of their parameter estimates are centered around the true parameter values. We document the estimators' bias and root-mean-squared errors (RMSE) in Table \ref{tab:simul performance}. Notably, the sieve estimators, while being more robust, tend to be no less efficient than the parametric estimators. The SML estimator's root-mean-squared errors (RMSE) are slightly larger than those of the ML$^*$ estimator for most parameters. The SLS and SGLS estimators---which perform very similarly because the measurement errors are uncorrelated---are slightly less efficient than the SML estimator, as expected.

\begin{table}[ht!]
\centering\footnotesize
\caption{\label{tab:simul performance}Finite sample performances of the estimators (Gaussian DGP)}
\begin{tabular}{ccccccc}
\toprule
&& ML & ML$^*$ & SML & SLS & SGLS \\
\midrule
$\alpha_{CC}$ & Bias& -0.0536 & -0.0041 & -0.0007 & -0.0018 & -0.0027\\
              & RMSE& 0.0775 & 0.0513 & 0.0520 & 0.0523 & 0.0523\\
\hline
$\alpha_{MM}$ & Bias& 0.0992 & 0.0044 & 0.0065 & 0.0031 & 0.0022\\
              & RMSE& 0.1077 & 0.0473 & 0.0491 & 0.0502 & 0.0489\\
\hline
$\beta_{C}$   & Bias& -0.0006 & -0.0007 & -0.0009 & -0.0009 & -0.0009\\
              & RMSE& 0.0416 & 0.0414 & 0.0427 & 0.0427 & 0.0427\\
\hline
$\beta_{M}$   & Bias& -0.0000 & -0.0001 & -0.0000 & -0.0000 & -0.0000\\
              & RMSE& 0.0398 & 0.0395 & 0.0397 & 0.0397 & 0.0397\\
\bottomrule
\end{tabular}
\end{table}

Now we consider DGPs for which the Gaussian model is moderately misspecified. We first generate $\{x_{1i},x_{2i}\}_{i=1}^n$ and $\{y_{1j},y_{2j}\}_{j=1}^n$ separately from the Gumbel copula with shape parameter values $1.3$ and $1.4$, respectively. These variables are uniformly distributed in $[0,1]$ by construction. Then we transform them into standard normally distributed variables:
$$  x_{Ci}=\Phi^{-1}(x_{1i}),\quad x_{Mi}=\Phi^{-1}(1-x_{2i}),\quad
    y_{Cj}=\Phi^{-1}(y_{1j}),\quad y_{Mj}=\Phi^{-1}(1-y_{2j}),$$
so that $(x_{Ci},x_{Mi})$ and $(y_{Cj},y_{Mj})$ are negatively correlated. Define matrices $x$ and $y$ by
$$  x:=\begin{bmatrix} x_{C1}&x_{M1} \\ \vdots&\vdots \\ x_{Cn}&x_{Mn} \end{bmatrix},\quad
    y:=\begin{bmatrix} y_{C1}&y_{M1} \\ \vdots&\vdots \\ y_{Cn}&y_{Mn} \end{bmatrix}.$$
The skill demand and supply bundles are standard normally distributed in this DGP so the ML$^*$ estimator is directly applicable without transformation, although it is misspecified because their joint distributions are not Gaussian. The Gumbel copula exhibits asymmetric tail dependence (the upper tail has stronger dependence than the lower tail), whereas the Gaussian copula has symmetric dependence. However, given the current parameter setup, the Gumbel copula does not drastically differ from the Gaussian copula. The production technology is specified as before. 

There exists no closed-form solution for the equilibrium assignment in this case. We therefore numerically solve the equilibrium matching through linear programming for each Monte Carlo sample. To do so, we first compute the pairwise surplus of each possible match between $x$ and $y$ and construct the surplus matrix $S$ whose $ij$ entry is the surplus generated by worker $i$ and firm $j.$ Let $\mathbb{I}_n$ denote an $n \times n$ identity matrix and $\textbf{1}_n$ be an $n \times 1$ vector of ones. Let $f$ be the vector obtained by stacking the columns of $S$. Then the solution ($q^*$) to the following linear programming problem provides the equilibrium assignment: 
\begin{equation}
    \max_{q} f'q,
    \quad{\rm s.t.}\
    \begin{bmatrix}
        \underbrace{A_1}_{n \times n^2} \\ \underbrace{A_2}_{n \times n^2}
    \end{bmatrix}q \le \underbrace{b}_{2n \times 1},\label{eq: linear programming}
\end{equation}
where $A_1 := \mathbb{I}_n \bigotimes \textbf{1}_n'$, $A_2:= \textbf{1}_n' \bigotimes \mathbb{I}_n,$ and $b:= \textbf{1}_{2n}.$ Reshaping $q^*$ into an $n \times n$ matrix gives the optimal transportation matrix, $T$. Then the optimal matching for $x$ is given by $y^*:=Ty$. The wage $w^*$ is computed from the solution to the dual problem \eqref{eq: linear programming}.\footnote{Even with the moderate sample size $n=3000$, the constraint matrix is enormous ($6000 \times 9,000,000).$ Solving the linear program \eqref{eq: linear programming} over many Monte Carlo samples is computationally demanding. We employ \textsc{Gurobi Optimizer 10.0} to solve it efficiently.} 

For measurement errors, we consider two different specifications. In the first case, the errors are independently drawn from a gamma distribution, $\Gamma(a,b)$, with $a = 1$, $b=2$. They are demeaned and scaled to have the same means and variances specified in the Gaussian DGP. Under this specification, the ML$^*$ and SML estimators are further misspecified as measurement errors are non-normal. In the second case, we generate the errors from a joint normal distribution in which the errors are correlated as follows:
$$  \begin{pmatrix} \varepsilon_w \\ \varepsilon_C \\ \varepsilon_M \end{pmatrix}
    \sim N\left(\begin{pmatrix} 0\\0\\0 \end{pmatrix},
                \begin{pmatrix} 2&1&1 \\ 1&1&0.5 \\ 1&0.5&1 \end{pmatrix}\right).$$
The ML$^*$ and SML estimators are still misspecified as measurement errors are correlated. The SLS estimator is consistent but not as efficient as the SGLS estimator, since it does not take the correlation structure of errors into account. The observable data $(y_i,x_i,w_i)_{i=1}^n$ are generated by adding the measurement errors to $w_i^*$, $y_{Ci}^*$, and $y_{Mi}^*$ respectively. 

\begin{table}[ht!]
\centering\footnotesize
\caption{\label{tab:simul performance gumbel}Finite sample performances of the estimators (Gumbel DGP)}
\begin{tabular}{cccccccccc}
\toprule
& &\multicolumn{4}{c}{Gamma errors} & \multicolumn{4}{c}{Joint Gaussian errors} \\
\cmidrule(lr){3-6} \cmidrule(lr){7-10}
&& ML$^*$ & SML & SLS & SGLS & ML$^*$ & SML & SLS & SGLS \\
\midrule
$\alpha_{CC}$ &Bias& -0.0608 & -0.0021 & -0.0015 & -0.0020 & -0.2032 & -0.0023 & 0.0008 & 0.0112 \\
              &RMSE& 0.3780 & 0.0538 & 0.0535 & 0.0538 & 0.2971 & 0.0914 & 0.0925 & 0.0809 \\
\hline
$\alpha_{MM}$ &Bias& -0.0550 & 0.0019 & 0.0024 & 0.0020 & -0.3339 & 0.0018 & 0.0034 & -0.0123 \\
              &RMSE& 0.4209 & 0.0512 & 0.0527 & 0.0513 & 1.6626 & 0.0886 & 0.0898 & 0.0827 \\
\hline
$\beta_{C}$   &Bias& 0.0055 & 0.0003 & 0.0003 & 0.0003 & -0.0015 & -0.0054 & -0.0054 & -0.0024 \\
              &RMSE& 0.1138 & 0.0406 & 0.0405 & 0.0406 & 0.1013 & 0.0762 & 0.0757 & 0.0760 \\
\hline
$\beta_{M}$   &Bias& -0.0056 & 0.0011 & 0.0011 & 0.0011 & -0.0327 & -0.0058 & -0.0056 & -0.0032 \\
              &RMSE& 0.1033 & 0.0398 & 0.0398 & 0.0399 & 0.1137 & 0.0735 & 0.0728 & 0.0677 \\
\bottomrule
\end{tabular}
\end{table}

The estimation results are provided in Table \ref{tab:simul performance gumbel}. In both cases, the ML$^*$ estimator is misspecified for the distributions of $X,$ $Y,$ and the measurement errors so that it performs the worst for all the parameters. It exhibits especially large biases and RMSE for complementarity parameters. Even for the linear productivity parameters, the ML$^*$ estimator shows much larger RMSEs than the sieve estimators. In contrast, all the sieve estimators work equally well for the linear coefficients. The SML estimator is misspecified for the distributions of measurement errors but it produces accurate estimates for the complementarity parameters $(\alpha_{CC}, \alpha_{MM})$ in both cases. The SLS and SGLS estimators perform similarly to the SML estimator when the measurement errors are drawn from the Gamma distributions. In the case of correlated errors, the SLS estimator performs similarly to the SML estimator. The SGLS estimator outperforms the other estimators as it takes into account the correlations between measurement errors, resulting in more efficient estimation.  

Lastly, we consider a DGP in which \cite{lindenlaub2017}'s model is more severely misspecified. Specifically, we draw $x$ and $y$ from finite Gaussian mixture distributions. Each mixture distribution has two Gaussian components with one-half weight for each. For both $x$ and $y$, the Gaussian components, $K_1$ and $K_2$, are specified as follows. 
$$  K_1 \sim N\left(\begin{pmatrix} 1\\1 \end{pmatrix},
                    \begin{pmatrix} 1&\rho \\ \rho&1\end{pmatrix}\right),\quad
    K_2 \sim N\left(\begin{pmatrix} -1\\-1 \end{pmatrix},
                    \begin{pmatrix} 1&-\rho \\ -\rho&1 \end{pmatrix}\right).$$
We set $\rho$ equal to $0.4$ for $x$ and $0.5$ for $y$. The equilibrium assignments and wages are solved via linear programming as before. We also generate the measurement errors from a joint Gaussian mixture distribution which has two components:
$$  M_1 \sim N\left(\begin{pmatrix} 1\\1\\1 \end{pmatrix},
                    \begin{pmatrix} 1&0.7&0.7 \\ 0.7&1&0.3 \\ 0.7&0.3&1 \end{pmatrix}\right),\quad
    M_2 \sim N\left(\begin{pmatrix} -3\\-3\\-3 \end{pmatrix},
                    \begin{pmatrix} 1&0.7&0.7 \\ 0.7&1&0.3 \\ 0.7&0.3&1 \end{pmatrix}\right).$$
In this case, both $(x,y)$ and measurement errors have bi-modal distributions that are far from a normal distribution.

\begin{table}[ht!]
\centering\footnotesize
\caption{\label{tab:simul performance mixture}Finite sample performances of the estimators (Gaussian mixture DGP)}
\begin{tabular}{cccccc}
\toprule
&& ML$^*$ & SML & SLS & SGLS \\
\midrule
$\alpha_{CC}$ & Bias& 0.2896 & -0.0014 & -0.0016 & -0.0017 \\
              & RMSE& 0.3653 & 0.0429 & 0.0425 & 0.0385 \\
\hline
$\alpha_{MM}$ & Bias& 0.2446 & -0.0017 & 0.0006 & -0.0027 \\
              & RMSE& 0.3584 & 0.0454 & 0.0431 & 0.0337 \\
\hline
$\beta_{C}$   & Bias& 0.7123 & -0.0007 & -0.0013 & 0.0008 \\
              & RMSE& 0.7204 & 0.0428 & 0.0426 & 0.0364 \\
\hline
$\beta_{M}$   & Bias& -0.1288 & -0.0001 & 0.0002 & -0.0011 \\
              & RMSE& 0.1592 & 0.0421 & 0.0419 & 0.0305 \\
\bottomrule
\end{tabular}
\end{table}

As the margins of $x$ and $y$ are not standard normal, the ML$^*$ estimator is not directly applicable. Therefore, we use the inverse transform method to convert $x$ and $y$ to standard normal variables for the ML$^*$ estimator. As the joint distribution after transformation is non-normal, ML$^*$ is misspecified. Our sieve estimators can be applied without this transformation so we use untransformed data for the sieve-based estimators. The estimates of technology parameters are reported in Table \ref{tab:simul performance mixture}. Not surprisingly, the ML$^*$ estimator delivers parameter estimates that are very different from the true values. On top of misspecification, the transformation procedure introduces additional bias as the actual assignments are determined on the original data. By contrast, the sieve estimators still perform well in this case. Estimators relying on fewer assumptions deliver more accurate estimates. The SML estimator produces the least precise estimates among the sieve estimators. The SGLS estimator incorporates the correlation structure among measurement errors so it possesses substantial efficiency gains compared to the SLS estimator. 

\begin{table}[ht!]
    \centering
    \caption{Average computation time (in seconds) and RMSE (in parentheses)}
    \label{tab:simul3D}
    \begin{tabular}{l c cc cc}
        \toprule
         & & \multicolumn{4}{c}{Time (RMSE)} \\
         \cmidrule(lr){3-6}
         & & \multicolumn{2}{c}{\textit{2 characteristics (2D)}} &\multicolumn{2}{c}{\textit{3 characteristics (3D)}} \\
         \cmidrule(lr){3-4} \cmidrule(lr){5-6}
        & Sieve Order ($k_n$) & $n=5,000$ & $n=10,000$ & $n=5,000$ & $n=10,000$ \\
        \midrule
        ML* & -- & 0.010 (0.034) & 0.016 (0.024) & 0.100 (0.047) & 0.157 (0.033) \\
        SML  & 2  & 0.034 (0.035) & 0.024 (0.024) & 0.169 (0.046) & 0.281 (0.032) \\
             & 3  & 0.066 (0.035) & 0.106 (0.025) & 0.459 (0.046) & 0.675 (0.032) \\
        SLS  & 2  & 0.058 (0.035) & 0.101 (0.025) & 0.375 (0.047) & 0.636 (0.033) \\
             & 3  & 0.133 (0.035) & 0.237 (0.025) & 1.568 (0.048) & 2.539 (0.033) \\
        SGLS & 2  & 0.171 (0.035) & 0.368 (0.024) & 0.738 (0.046) & 1.254 (0.032) \\
             & 3  & 0.320 (0.035) & 0.631 (0.025) & 3.221 (0.046) & 5.196 (0.032) \\
        \bottomrule
    \end{tabular}

    \vspace{1mm}
    {\footnotesize \textit{Note:} Times are measured in MATLAB on a MacBook Pro with Apple Silicon (M1, 8-core CPU).}
\end{table}

Lastly, we investigate the computational scalability of our sieve estimators in settings where worker and job attributes have more than two dimensions. In additional simulations, we compare computation times of the parametric ML$^*$ estimator and sieve-based estimators in Gaussian designs with two (2D) and three characteristics (3D). The results are summarized in Table \ref{tab:simul3D}. As expected, the ML$^*$ estimator is the fastest in this Gaussian benchmark. One more characteristic increases the computing time for ML$^*$ roughly tenfold. The sieve-based methods are only moderately slower in absolute terms. For instance, with three characteristics, $n=10{,}000,$ and the Bernstein polynomial degree of 2 $(k_n=2$) for each attribute, SML requires 0.281 seconds versus 0.157 seconds for ML*, while delivering virtually identical RMSE. In all cases reported, computation times remain well below one second for SML and below a few seconds for the more demanding SGLS specifications. The additional dimension tends to increase the computation time for the sieve estimators no more than tenfold, and often much less. Furthermore, relatively low sieve orders are sufficient to match the performance of the correct parametric model. With $k_n=2$, all sieve estimators attain RMSEs that are essentially indistinguishable from those of ML$^*$. This suggests that, in practice, choosing modest sieve orders keeps the effective parameter dimension small, while retaining flexibility.

Regarding scalability beyond 3D, our sieve estimators can be computationally manageable for moderate numbers of characteristics (e.g., 3--5) and low sieve orders, as the simulations illustrate in the most demanding SGLS configurations. In applications with richer covariate sets, it is natural to combine our method with standard dimension-reduction tools (e.g., economically motivated indices or principal components) and to keep sieve orders small. The sieve estimators also offer a great ``implementation'' advantage. Computational burden in practice is not only about CPU time but also about implementation cost. In ML estimation, the log-likelihood depends on the parameters in a complex nonlinear way, so deriving analytic expressions of the gradient and Hessian becomes cumbersome and error-prone for higher dimensions. By contrast, for our sieve-based estimators, the objective functions are much simpler in the sieve coefficients: their gradients and Hessians have straightforward closed-form expressions that can be implemented with relatively little coding effort. This substantially reduces the ``programmer time'' required to obtain a stable and efficient implementation in higher-dimensional specifications.

Our simulation exercises provide evidence that the Gaussian model can be misleading when the model is misspecified. Even with moderate misspecification, the ML estimator does not produce reliable estimates. Furthermore, the actual impact of technological shifts on wage distribution may differ from the prediction of the Gaussian model. \citet{lindenlaub2017} shows that (i) wage distributions are positively skewed for any pairs of $\alpha_{CC}$ and $\alpha_{MM}$, (ii) the variance of the wage distribution increases as cognitive or manual skill complementarity increases, and (iii) wage skewness is minimized when $\alpha_{CC}=\alpha_{MM}$. However, in our simulations using non-normal distributions (detailed in Appendix \ref{appen: skewdisp}), the resulting wage distribution's skewness does not reach its minimum when $\alpha_{CC}=\alpha_{MM}$. On the contrary, our sieve estimators do not suffer from these problems. Therefore, our estimators can be more generally applicable regardless of underlying distributions of $x$, $y$, and measurement errors with no need for transformation. They also show excellent finite sample performances and remain computationally feasible for large samples and for settings with more than two attributes.  

\section{\label{sec: Matchingemp}Empirical application to U.S. worker-job matching}

We revisit the dataset constructed by \cite{lindenlaub2017} to learn how production technology in the U.S.\ has evolved. We estimate the production technology parameters in the model using the dataset and sieve estimators. The National Longitudinal Survey of Youth (NLSY) data and U.S. Department of Labor Occupational Characteristics Database (O{*}NET) are used to construct workers' cognitive and manual skills as well as the occupational skill requirements of firms. To assess the effect of technological changes on wage inequality, we compare estimation results based on two cohorts: the first cohort commencing in 1979 (referred to as NLSY79) and the second commencing in 1997 (referred to as NLSY97). Following Lindenlaub's main specification, we focus on employed workers aged 27 to 29 during the years 1990-91 and 2009-10, sourced from the NLSY79 and NLSY97 cohorts, respectively.\footnote{The dataset excludes military samples and oversamples of special demographic/racial groups to give primary focus on the core sample of the NLSY.} The wage, $w$, is defined as the CPI-adjusted hourly rate.  

Firms' skill demands, $(y_C,y_M),$ are constructed from the O{*}NET, which contains information on skill requirements for each occupation. \citet{sanders2014} classifies occupational skill requirements into cognitive and manual categories and constructs two task scores for over 400 occupations. These scores are employed to obtain skill demands at the occupation level.\footnote{For instance, the occupation `dancer' has a normalized cognitive score ($y_C$) of $0.34$ and a normalized manual score ($y_M$) of 1. By contrast, `physicist' has a skill demand bundle of $(y_C=1,y_M=0.11).$ We use the normalized task scores for illustration purposes following \cite{lindenlaub2017}. The original supports of worker skills and firms' skill demands are provided in Table \ref{tab: statistics-matching}.} The skill bundle $(x_C,x_M)$ is constructed using survey responses in the NLSY on education and training that are predetermined prior to employment. Given college education, apprenticeships, government training degrees, and occupational training history, workers are qualified for specific occupations from which their skill bundles are determined.\footnote{For example, a worker who studied economics at a university is qualified for the `economist' job. Then the worker possesses a normalized skill bundle of $\left(x_{C}=0.615,x_{M}=0.034\right)$, that is required to be an economist.} We argue that measurement errors in $y_i$ arise primarily from coarsening and aggregation of O*NET descriptors and survey noise. As the workers' skills are independent of their current (realized) occupation and skill requirements of each occupation are not tailored to individual workers, the exogeneity of measurement errors we impose for estimation is an empirically reasonable approximation in our setting.

Table \ref{tab: statistics-matching} presents summary statistics of workers' skills and firms' skill demands. In 1990/91, workers had higher cognitive skills on average than manual skills, and firms also required more cognitive skills than manual skills. Two decades later, workers had increased cognitive skills and decreased manual skills on average compared to 1990/91, with firms also showing a similar trend. The skill correlation ($\rho_x$) shifted from $-0.40$ to $-0.52$, indicating increased worker specialization. In contrast, the productivity correlation ($\rho_y$) remained stable at $-0.49$. Initially, jobs were more specialized than workers, but skill supply caught up, resulting in slightly greater worker specialization in 2009/10.

\begin{table}[t!]
\caption{\label{tab: statistics-matching}Summary statistics of worker skills ($x$) and firms' skill demands ($y$)}
\centering{}%
\begin{tabular}{cr@{\extracolsep{0pt}.}lr@{\extracolsep{0pt}.}lr@{\extracolsep{0pt}.}lr@{\extracolsep{0pt}.}lr@{\extracolsep{0pt}.}lr@{\extracolsep{0pt}.}lr@{\extracolsep{0pt}.}lr@{\extracolsep{0pt}.}l}
\hline 
 & \multicolumn{8}{c}{1990/91 ($n=2984$)} & \multicolumn{8}{c}{2009/10 ($n=4495$)}\tabularnewline
\cmidrule(lr){2-9} \cmidrule(lr){10-17}
 & \multicolumn{2}{c}{$x_{C}$} & \multicolumn{2}{c}{$x_{M}$} & \multicolumn{2}{c}{$y_{C}$} & \multicolumn{2}{c}{$y_{M}$} & \multicolumn{2}{c}{$x_{C}$} & \multicolumn{2}{c}{$x_{M}$} & \multicolumn{2}{c}{$y_{C}$} & \multicolumn{2}{c}{$y_{M}$}\tabularnewline
\hline 
{Mean} & 0&3596 & -0&2912 & 0&0135 & -0&1189 & 0&5667 & -0&6601 & 0&0468 & -0&2509\tabularnewline
{SD} & 0&7423 & 0&9923 & 0&8490 & 1&0240 & 0&7556 & 0&8358 & 0&9280 & 0&9656\tabularnewline
{Min} & -2&0595 & -1&7004 & -2&0622 & -1&6949 & -2&3019 & -1&8116 & -2&5200 & -1&6597\tabularnewline
{Max} & 2&1649 & 2&1855 & 2&0925 & 2&1895 & 1&9160 & 2&1838 & 3&0504 & 2&1351\tabularnewline
\hline 
\end{tabular}
\end{table}

\cite{lindenlaub2017} transforms each element of $x$ and $y$ into a standard Gaussian variable, and their dependence is modeled using the Gaussian copula. However, while their marginal distributions are standard normal, the transformed variables are not guaranteed to have a joint normal distribution. Let $\tilde{x}$ and $\tilde{y}$ be Gaussian-transformed $x$ and $y$ respectively. As shown in Figure \ref{fig:dist-transformed}, the joint distributions of $\tilde{x}$ and $\tilde{y}$ are not normal. In particular, the joint density of $\tilde{y}$ is multi-modal in both periods. 

\begin{figure}[t!]
\begin{centering}
\includegraphics[width=1\textwidth]{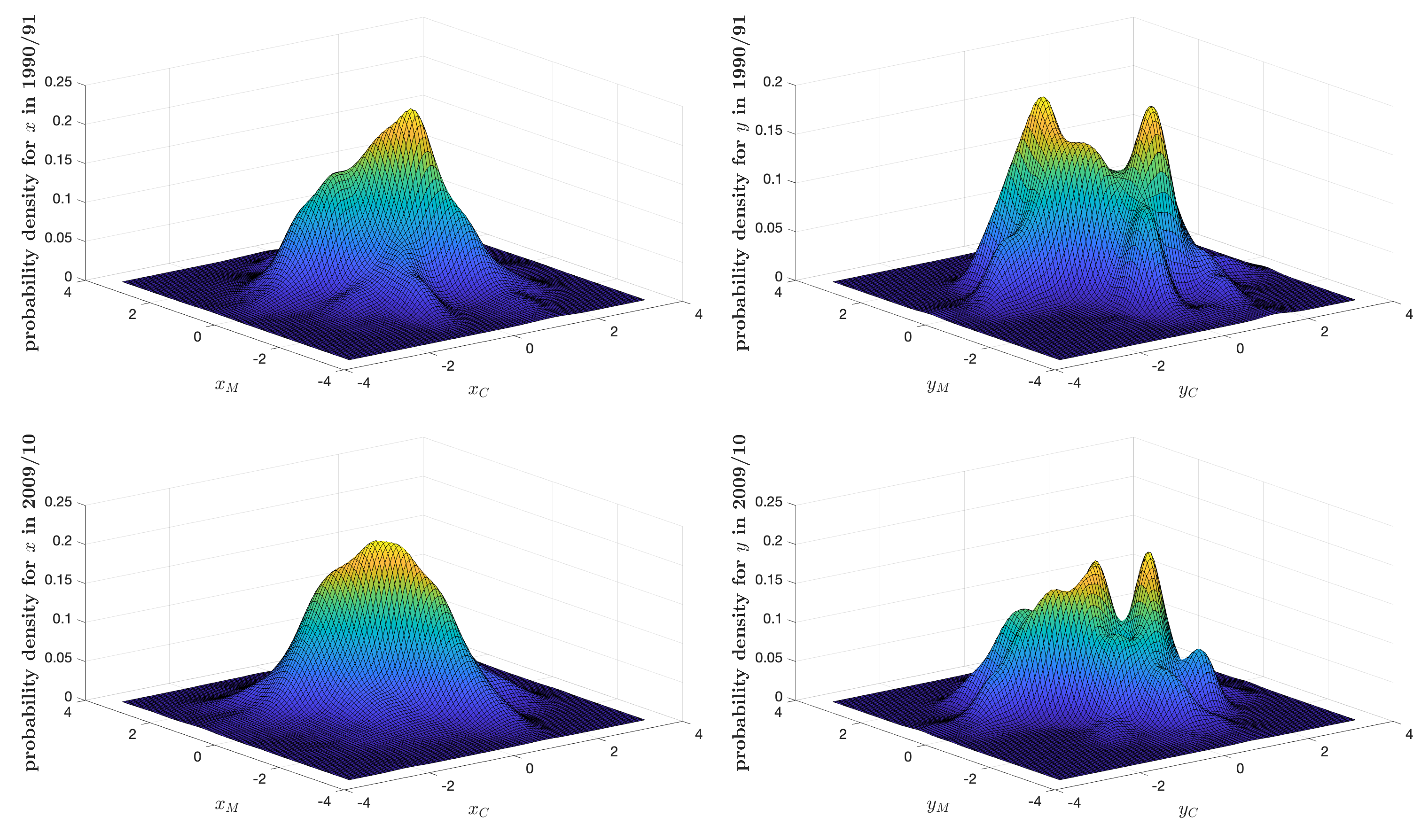}
\par\end{centering}
\caption{\label{fig:dist-transformed}Joint densities of $\tilde{x}$ and $\tilde{y}$ (transformed data)}
\end{figure}

We employ Mardia's test \citep{mardia1970} to formally test the joint normality of $\tilde{x}$ and $\tilde{y}$. We first create an $n \times n$ matrix for $\tilde{x}$:
\[
    C=\left(c_{ij}\right)=x^{*}S^{-1}\left(x^{*}\right)',
\]
where the $i$-th row of $x^{*}$ is $x_{i}^{*}=\tilde{x}_{i}-\sum_{i=1}^{n}\tilde{x}_{i}/n$, and define multivariate measures of skewness and kurtosis as follows:
\[
    b_{1}=\frac{1}{n^{2}}\sum_{i=1}^{n}\sum_{j=1}^{n}c_{ij}^{3},\quad
    b_{2}=\frac{1}{n}\sum_{i=1}^{n}c_{ii}^{2}.
\]
Under bivariate normality, the limiting distribution of $\frac{nb_{1}}{6}$ is a chi-square distribution with $d\left(d+1\right)\left(d+2\right)/6$ degrees of freedom and the limiting distribution of $\frac{\sqrt{n}\left(b_{2}-d\left(d+2\right)\right)}{\sqrt{8d\left(d+2\right)}}$ is the standard normal distribution where $d$ is the dimensionality of $\tilde{x}$. We conduct the same procedure for $\tilde{y}$. Table \ref{tab: mardia-transformed} shows the test statistics. The test strongly rejects the bivariate normality of $\tilde{x}$ and $\tilde{y}$ in both periods. The normality assumption imposed in Lindenlaub's model is not satisfied even after transforming $x$ and $y$. Therefore, our semiparametric approach is more appropriate in this case.

\begin{table}[tbh]
\caption{\label{tab: mardia-transformed}Mardia's multivariate normality test statistics (p-values in parentheses)}
\centering{}%
\begin{tabular}{cr@{\extracolsep{0pt}.}lr@{\extracolsep{0pt}.}lr@{\extracolsep{0pt}.}lr@{\extracolsep{0pt}.}l}
\hline 
 & \multicolumn{4}{c}{1990/91 ($n=2984$)} & \multicolumn{4}{c}{2009/10 ($n=4495$)}\tabularnewline
\cmidrule(lr){2-5} \cmidrule(lr){6-9}
 & \multicolumn{2}{c}{$\tilde{x}$} & \multicolumn{2}{c}{$\tilde{y}$} & \multicolumn{2}{c}{$\tilde{x}$} & \multicolumn{2}{c}{$\tilde{y}$}\tabularnewline
\hline 
Skewness & 4&58 (0.333) & 100&09 (0.000) & 16&34 (0.003) & 145&14 (0.000)\tabularnewline
Kurtosis & 4&44 (0.000) & 0&29 (0.774) & 14&42 (0.000) & 1&98 (0.048)\tabularnewline
\hline 
\end{tabular}
\end{table}

We estimate the production technology in each period separately. As in the simulations, we shut down between-task complementarities ($A$ is diagonal) following \cite{lindenlaub2017}. As described in the simulation section, the corrected ML procedure (ML$^*$) modifies \cite{lindenlaub2017}'s MLE to use the latent productivity correlation in place of the observed one. Here we further generalize ML$^*$ to accommodate the case where the true skill requirements $\tilde{y}^*$ have variances not equal to $1$ (since the inverse transform method converts the measurement error contaminated $\tilde{y}$, not $\tilde{y}^*$) and allow the Gaussian measurement errors to be correlated with each other. We then estimate the technology parameters using our sieve estimators (SML, SLS, and SGLS) with Bernstein polynomial basis functions of degree $3$, which perform the best in terms of information criteria and model fit. We compare our semiparametric estimates of the technology parameters to Lindenlaub's original estimates and the ML$^{*}$ estimates.

\begin{table}[t!]
\caption{\label{tab: est-matching-trans}Estimates of production technology parameters on transformed data}
\centering{}%
{\footnotesize
\begin{tabular}
{cr@{\extracolsep{0pt}.}lr@{\extracolsep{0pt}.}lr@{\extracolsep{0pt}.}lr@{\extracolsep{0pt}.}lr@{\extracolsep{0pt}.}lr@{\extracolsep{0pt}.}lr@{\extracolsep{0pt}.}lr@{\extracolsep{0pt}.}lr@{\extracolsep{0pt}.}lr@{\extracolsep{0pt}.}l}
\hline 
 & \multicolumn{10}{c}{1990/91} & \multicolumn{10}{c}{2009/10}\tabularnewline
\cmidrule(lr){2-11} \cmidrule(lr){12-21}
& \multicolumn{2}{c}{ML} &\multicolumn{2}{c}{ML*} & \multicolumn{2}{c}{SML} & \multicolumn{2}{c}{SLS} & \multicolumn{2}{c}{SGLS} &\multicolumn{2}{c}{ML} &\multicolumn{2}{c}{ML*} & \multicolumn{2}{c}{SML} & \multicolumn{2}{c}{SLS} & \multicolumn{2}{c}{SGLS}\tabularnewline
\hline 
$\alpha_{CC}$ &  0&203 & 0&765 & 0&454 & 0&000  & 0&000  & 0&739 &  1&119 &  2&048 & 2&293 & 2&290\tabularnewline
              & (0&342)&(0&574)&(0&009)& (0&000)&(0&000) &(0&198)& (0&408)& (0&036)&(0&256)&(0&265)\tabularnewline
$\alpha_{MM}$ &  0&479 & 1&270 & 1&422 & 1&084  & 0&856  & 0&055 &  0&486 &  0&237 & 1&033 & 0&291\tabularnewline
              & (0&175)&(0&148)&(0&031)& (0&113)&(0&098) &(0&154)& (0&633)& (0&006)&(0&215)&(0&059)\tabularnewline
$\beta_{C}$   &  1&686 & 1&711 & 1&692 & 1&719  & 1&585  & 2&203 &  2&208 &  2&115 & 2&063 & 2&198\tabularnewline
              & (0&143)&(0&589)&(0&068)& (0&434)&(0&416) &(0&152)& (0&540)& (0&068)&(0&589)&(0&534)\tabularnewline
$\beta_{M}$   & -0&421 & -0&392 &-0&374 & -0&382 &-0&388  & 0&210 &  0&243 &  0&198 & 0&180 & 0&327\tabularnewline
              & (0&141)&(0&406)&(0&068)& (0&329)&(0&309) &(0&152)& (0&731)& (0&076)&(0&545)&(0&532)\tabularnewline
\hline 
\multicolumn{21}{l}{{\footnotesize Standard errors in parentheses. ML indicates the original estimates in \cite{lindenlaub2017}.}}\tabularnewline
\end{tabular}}
\end{table}

The estimation results are provided in Table \ref{tab: est-matching-trans}. All the models clearly show a substantial shift in the relative importance between manual and cognitive tasks over the two decades. Our results are qualitatively consistent with Lindenlaub's but quantitatively very different. In 1990/91, the estimated complementarity in manual tasks was much larger than in cognitive tasks. The Gaussian model (ML$^*$) indicates that the complementarity in manual tasks is roughly 1.7 times as large as that of cognitive tasks. When we dispense with the normality assumption on skill demand and supply, the ratio becomes larger than 3. When we further generalize the model by removing the Gaussian assumption on measurement errors, the complementarity in cognitive tasks shrinks close to 0 (but significantly larger than 0), whereas that of manual tasks is still close to 1. The estimates of linear productivity coefficients $\beta_C$ and $\beta_M$ are similar across all the specifications as shown in the simulations.

\begin{figure}[b!]
\begin{centering}
\includegraphics[width=1\textwidth]{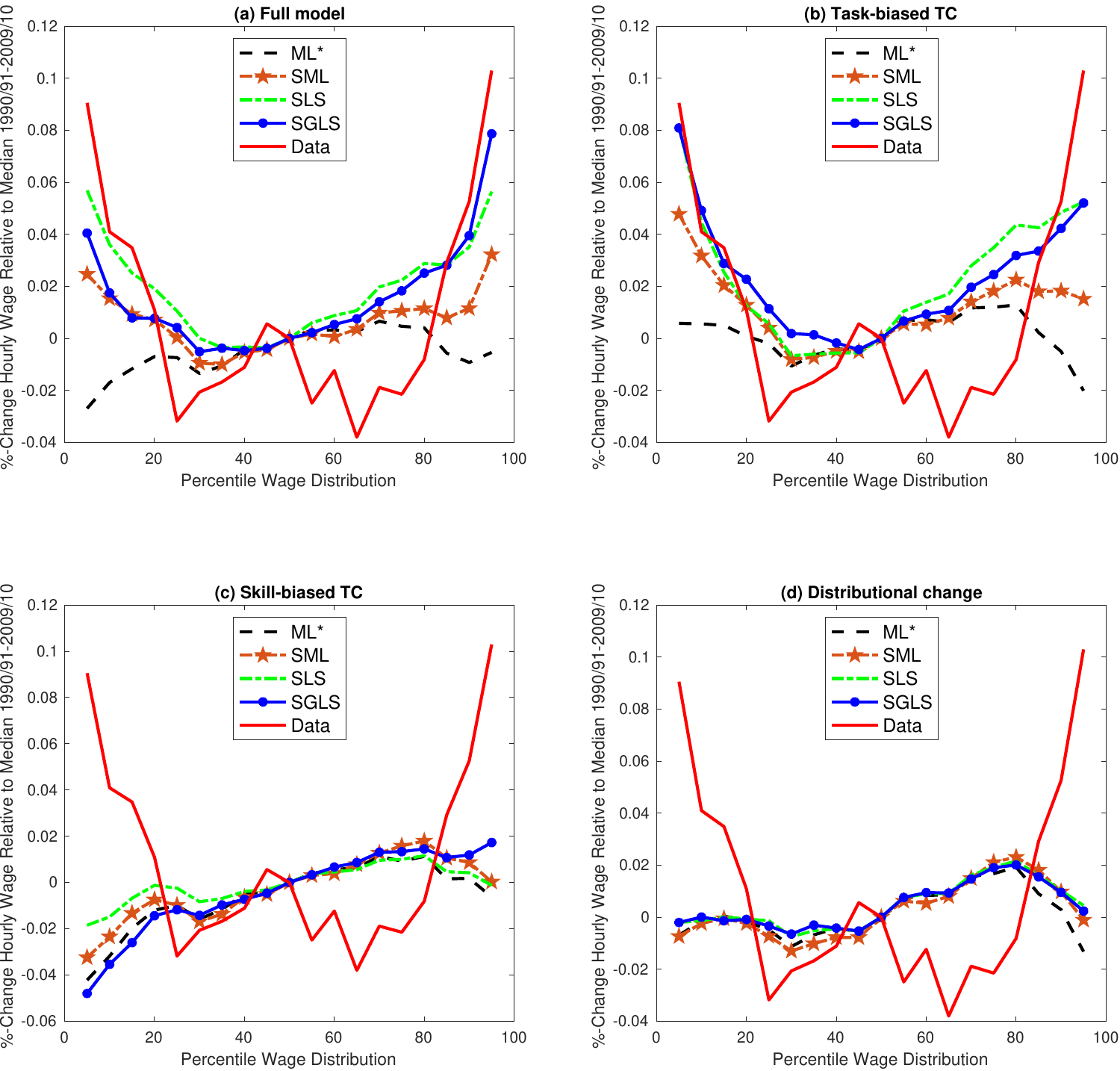}
\par\end{centering}
\caption{\label{fig:wagepol}Actual and model predicted wage polarization (transformed data)}
\end{figure}

This pattern reverses over the following 20 years. All the complementarity estimates for 2009/10 indicate a substantial increase in the complementarity between cognitive worker and job attributes, whereas the complementarity in manual tasks sharply decreased. In the Gaussian model (ML$^*$), the ratio of estimated complementarities in manual tasks to cognitive tasks ($\frac{\alpha_{MM}}{\alpha_{CC}}$) is around $0.43$, which is similar to the estimated value by SLS. However, SML and SGLS deliver much smaller values close to $0.1$. The linear coefficients are similar across specifications. These patterns imply substantial changes in the relative complementarities across tasks because of technological advances. Lindenlaub describes this phenomenon as ``task-biased technological change in favor of cognitive tasks''.
The cognitive dimension became much more important in labor market sorting. Our semiparametric models suggest that the ``task-biased technological change'' favoring cognitive tasks in the last two decades may have been much larger than previously found. The increases in $\beta_C$ and $\beta_M$ indicate that both cognitive and manual skill productivity have risen. However, the estimated manual skill productivity in both periods is insignificant in most specifications. Therefore, following Lindenlaub's description, we can conclude that the U.S. economy has also experienced ``skill-biased technological change'' in favor of cognitive skills. 

We now investigate the effect of estimated technological changes on wage inequality. Wage inequality in the U.S. labor market until the late 2000s is well characterized by \textit{wage polarization}, which is defined as stronger wage growth in the bottom and upper tails relative to the median. The red solid line in Figure \ref{fig:wagepol} plots how wages changed relative to the median wage between 1990/91 and 2009/10 by wage percentile, implying that the U.S. labor market experienced a spike in the upper-tail wage inequality, while the lower-tail inequality declined. This phenomenon is surprisingly well predicted by our models, as shown in Figure \ref{fig:wagepol} (a). All the semiparametric models predict substantial wage polarization once the estimated parameter values are fed in. By contrast, the Gaussian model fails to account for wage polarization in both tails. We can also observe that the model fit improves as the model becomes more flexible.

\begin{figure}[b!] 
\begin{centering}
\includegraphics[width=1\textwidth]{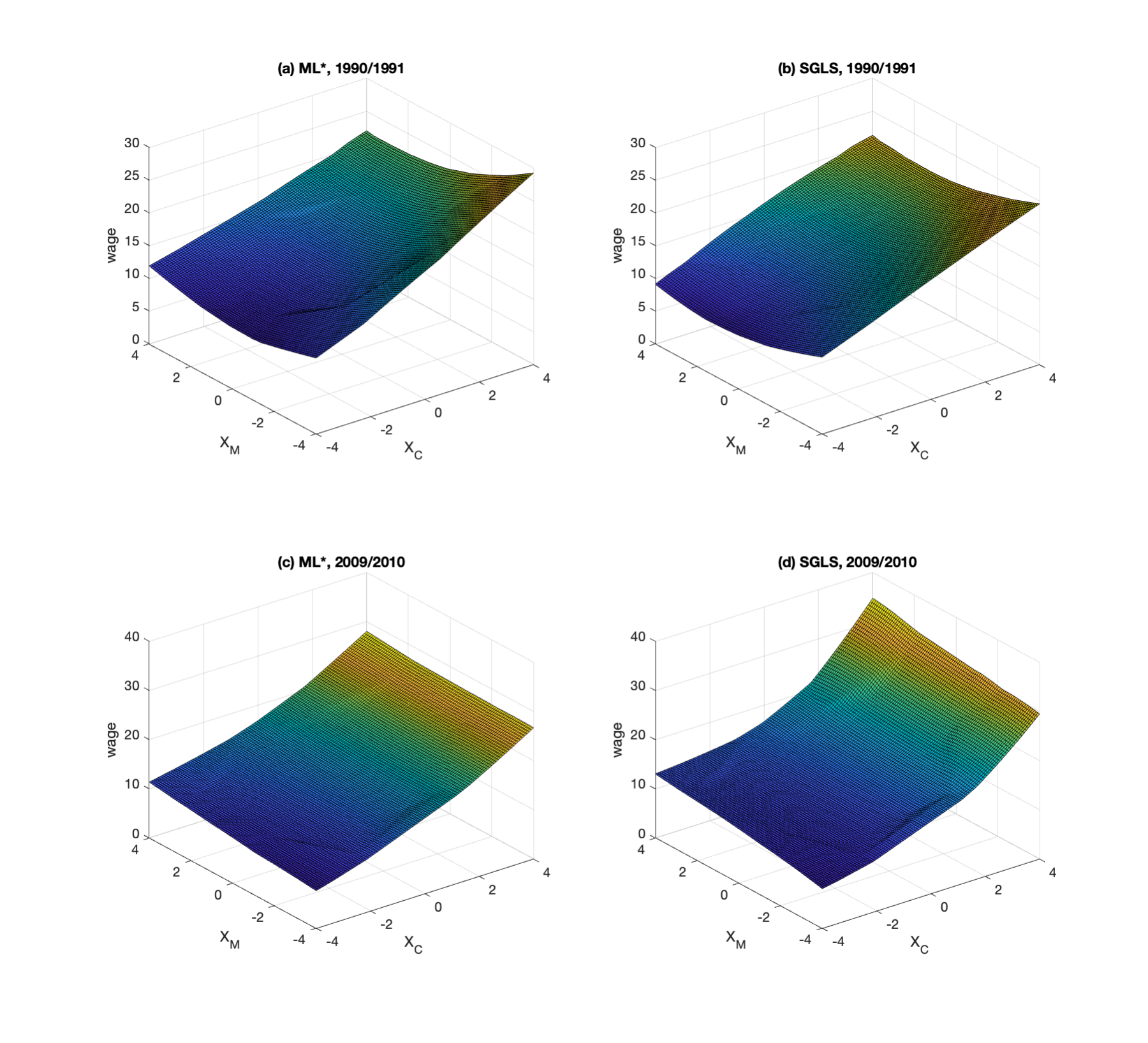}
\par\end{centering}
\caption{\label{fig:wage3d}Estimated wage functions (transformed data)}
\end{figure}

To further explore why the matching model requires greater flexibility to account for wage polarization, we compare the curvature of estimated wage functions across different models in Figures \ref{fig:wage3d}--\ref{fig:wage curvature}. In 1990/91, all models produce almost linear wage functions, indicating a relatively uniform relationship between wages and cognitive skills. However, in 2009/10, our semiparametric models predict a significantly steeper curvature, particularly at high cognitive skill levels, whereas the Gaussian model still generates a more linear wage function. The Gaussian model constrains the wage function to a quadratic form of standard normal variables, limiting its shape to a low-degree polynomial. In contrast, our models do not impose such constraints, allowing for greater flexibility in curvature that better fits the data. Notably, our most flexible model predicts a sharply increasing slope in the wage function for 2009/10, which is relatively flat at low cognitive skill levels and very steep at high skill levels, generating substantial wage polarization.

\begin{figure}[t!]
\begin{centering}
\includegraphics[width=1\textwidth]{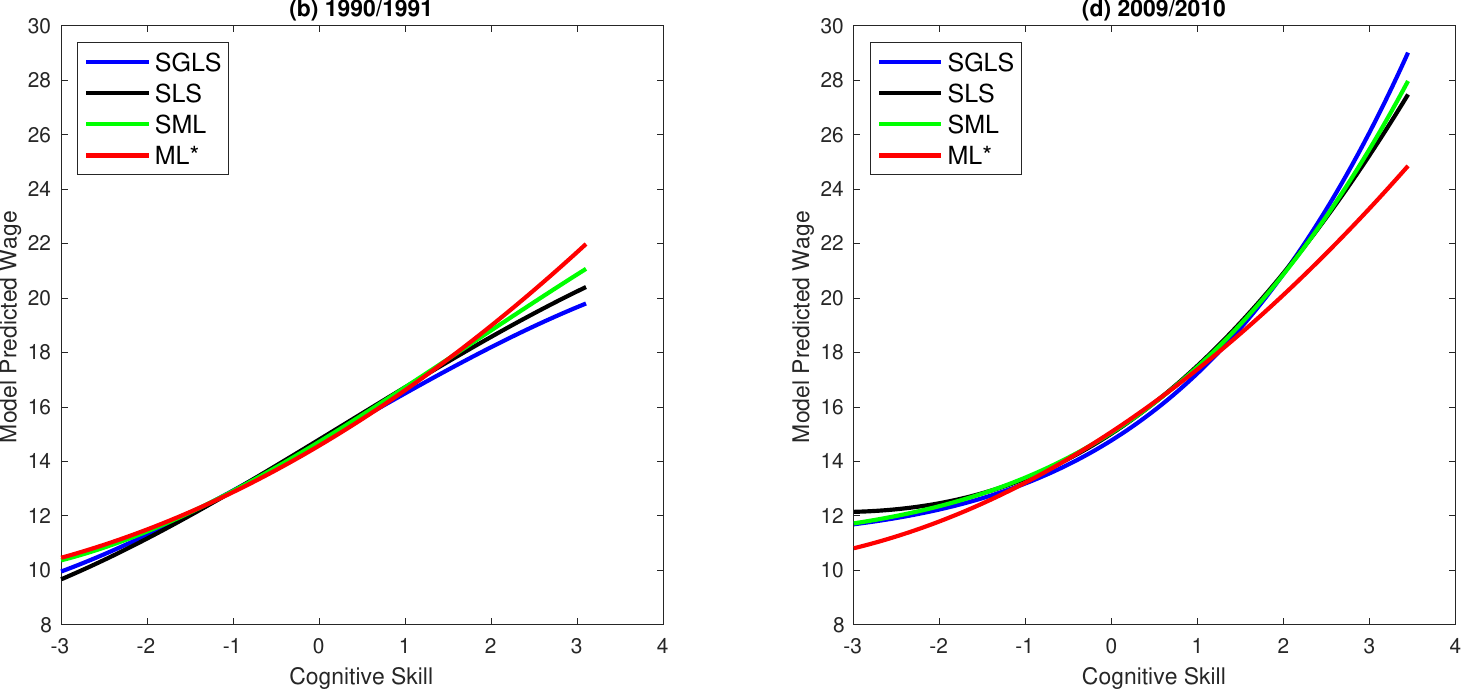}
\par\end{centering}
\caption{\label{fig:wage curvature}Predicted wage with respect to cognitive skill (transformed data)}
\end{figure}

To understand the driving forces behind wage polarization, we isolate the effects of technological and distributional changes in Figure \ref{fig:wagepol}. We only keep task-biased technological change (shutting down changes in linear productivity coefficients) in panel (b), skill-biased technological change (shutting down changes in complementarity parameters) in (c), and distributional change (shutting down both changes in linear productivity and complementarity parameters) in (d). We find that task-biased technological change explains wage polarization remarkably well. In particular, all three semiparametric models exhibit an excellent fit in the lower tail in panel (b), while the Gaussian model shows only a slight decline in lower tail inequality. In contrast, skill-biased technological change exacerbates wage inequality in the lower tail as shown in panel (c). The distributional change has a negligible impact on wage inequality. In summary, despite their parsimony, our matching models effectively account for the evolution of wage inequality over the past 20 years in the U.S., with task-biased technological change being the primary driver of the observed pattern.

\begin{figure}[b!]
\begin{centering}
\includegraphics[width=1\textwidth]{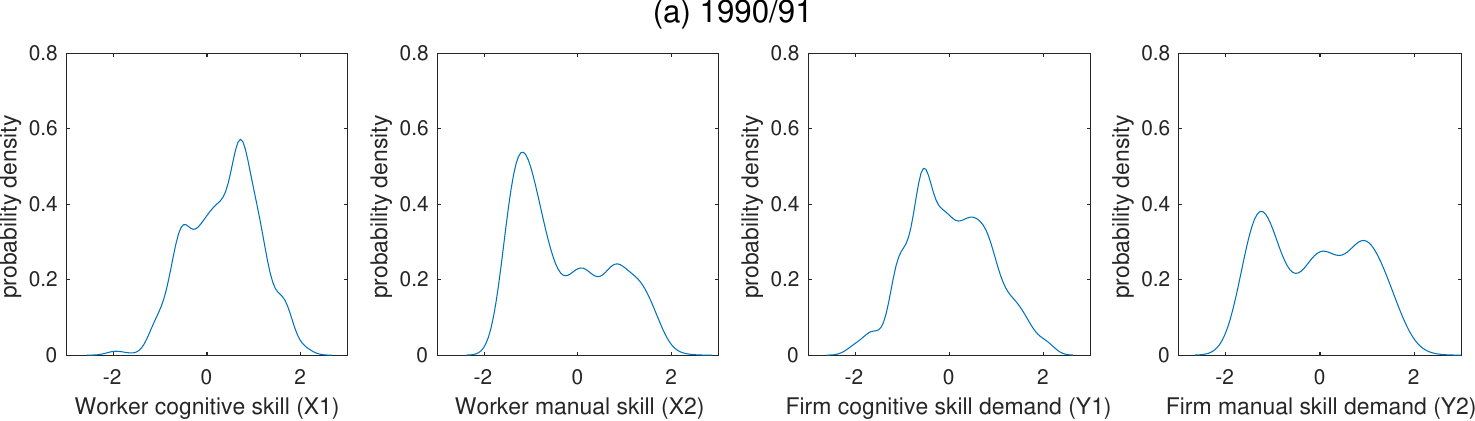}
\medskip

\includegraphics[width=1\textwidth]{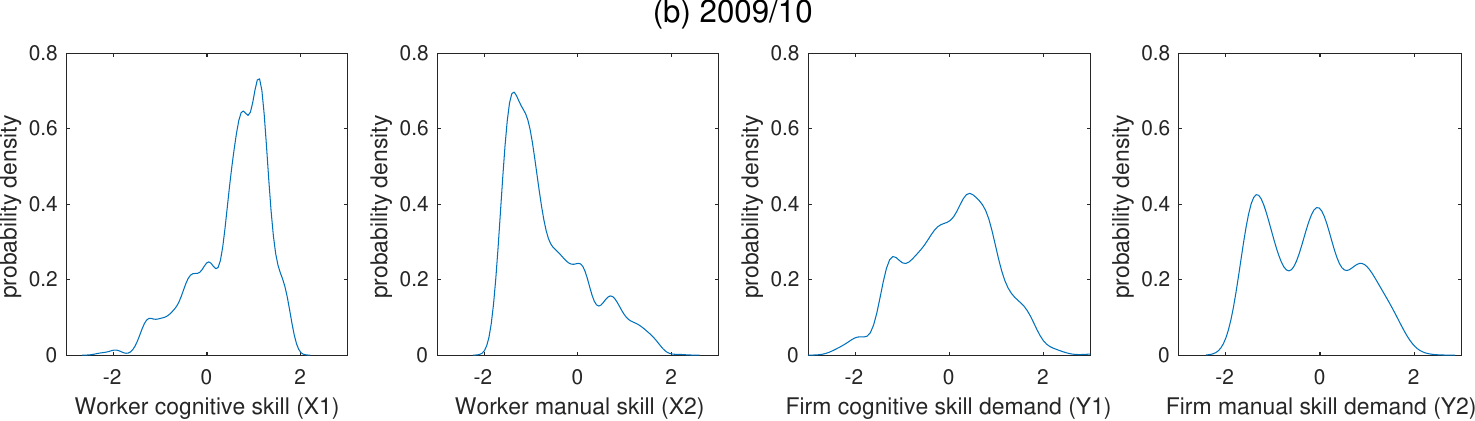}
\par\end{centering}
\caption{\label{fig:marginal dist}Marginal distributions of skill supply and demand}
\end{figure}

Lastly, we estimate our semiparametric models on the original data. Unlike the Gaussian model, our models can be applied directly to the data without any transformation. As illustrated in Figure \ref{fig: lindenlaub}, the matching occurs based on the original distributions. Hence, transforming marginal distributions to standard normal can lead to a different solution from $T(x).$ The marginal distributions of workers' skill supply $x$ and firms' skill demand $y$ in Figure \ref{fig:marginal dist} are skewed and multi-modal, quite different from standard normal. Therefore, it is crucial to investigate the robustness of the Gaussian model using the original data. We report the estimated parameters in Table \ref{tab: est-matching-original}, using Bernstein polynomial basis functions of degree 4 to accommodate the less well-behaved original data. The estimated parameters reveal similar patterns across our models. Notably, the complementarity in cognitive tasks increased significantly from near 0 in 1990/91 to around 6 in 2009/10, while the complementarity in manual tasks decreased from near 2 to almost 0. Additionally, linear skill productivity improved, with a larger increase in cognitive skill productivity. These findings confirm that the U.S. economy experienced substantial task-biased and skill-biased technological changes favoring cognitive skills over the two decades.

\begin{table}[tbh]
\caption{\label{tab: est-matching-original}Estimates of production technology parameters on original data}
\centering{}%
\footnotesize
\begin{tabular}
{cr@{\extracolsep{0pt}.}lr@{\extracolsep{0pt}.}lr@{\extracolsep{0pt}.}lr@{\extracolsep{0pt}.}lr@{\extracolsep{0pt}.}lr@{\extracolsep{0pt}.}l}
\hline 
 & \multicolumn{6}{c}{1990/91} & \multicolumn{6}{c}{2009/10}\tabularnewline
 \cmidrule(lr){2-7} \cmidrule(lr){8-13}
 & \multicolumn{2}{c}{SML} & \multicolumn{2}{c}{SLS} & \multicolumn{2}{c}{SGLS} & \multicolumn{2}{c}{SML} & \multicolumn{2}{c}{SLS} & \multicolumn{2}{c}{SGLS}\tabularnewline
\hline 
$\alpha_{CC}$ & 0&000  & 0&000 &  0&000 &  6&032 & 5&784 & 6&471\tabularnewline
              &(0&000) &(0&000)& (0&000)& (0&110)&(0&184)&(0&196)\tabularnewline
$\alpha_{MM}$ & 1&947  & 2&292 &  2&234 &  0&000 & 1&020 & 0&003\tabularnewline
              &(0&060) &(0&081)& (0&079)& (0&000)&(0&024)&(0&000)\tabularnewline
$\beta_{C}$   & 2&238  & 2&246 &  2&108 &  3&395 & 3&299 & 3&707\tabularnewline
              &(0&083) &(0&225)& (0&220)& (0&101)&(0&396)&(0&292)\tabularnewline
$\beta_{M}$   &-0&341  & -0&441 & -0&317 & 0&254 & 0&384 & 0&440\tabularnewline
              &(0&073) &(0&211)& (0&213)& (0&097)&(0&333)&(0&239)\tabularnewline
\hline 
\multicolumn{13}{l}{{\footnotesize{}Standard errors in parentheses}}\tabularnewline
\end{tabular}
\end{table}

\begin{figure}[h!]
\begin{centering}
\includegraphics[width=1\textwidth]{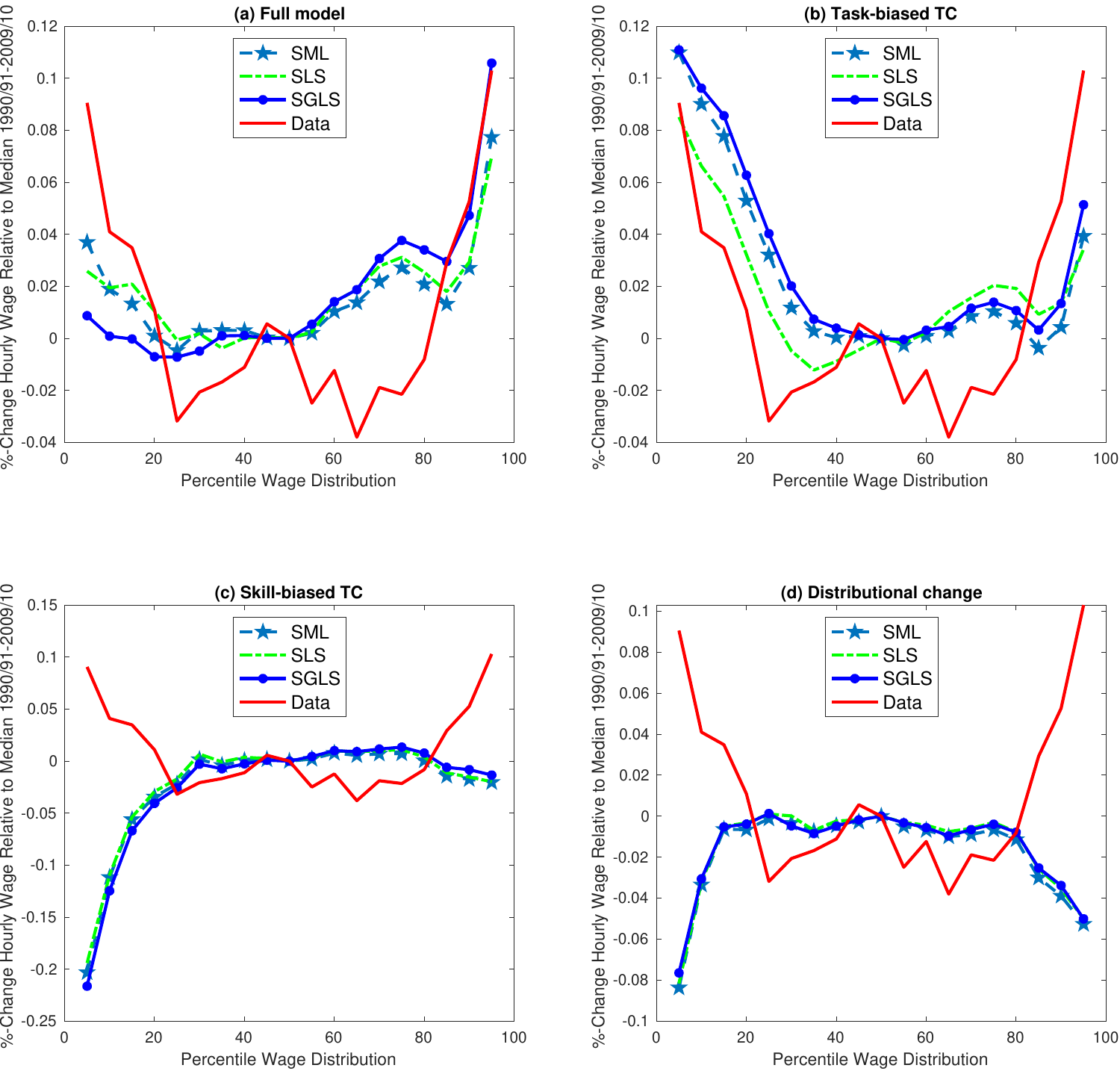}
\par\end{centering}
\caption{\label{fig:wagepol-UT}Actual and model predicted wage polarization (original data)}
\end{figure}

The estimated models on the original data effectively capture the patterns of wage polarization, particularly in the upper tail, as shown in Figure \ref{fig:wagepol-UT}. While the models slightly under-predict wage polarization in the lower tail, they confirm that task-biased technological change was the primary driver of wage polarization. In the absence of skill-biased technological change, the model shows a significant relative wage increase in the lower tail. However, skill-biased technological change had a negative impact on wage inequality, exacerbating lower tail inequality. Distributional change improved upper tail inequality but worsened lower tail inequality.

In summary, the estimated semiparametric models on the original data exhibit similar patterns to those on the Gaussian transformed data, with task- and skill-biased technological changes being more pronounced. Our models demonstrate a remarkable fit for wage polarization, highlighting the substantial changes in production technology in the U.S. over the past two decades. This exercise showcases the versatility and effectiveness of our semiparametric models and sieve-based estimators, which can accommodate any underlying joint distributions of skill supply and demand without requiring data transformation or distributional assumptions. Moreover, our approach builds upon standard sieve-based estimators, which have been shown to achieve semiparametric efficiency and are easy to implement in practice.

%We could think of possible extensions. I derive the identification result for the technology with between-task complementarities, but the estimation is based on the bilinear technology to compare with the results in \citet{lindenlaub2017}. Also, there might be another heterogeneity, e.g. interpersonal skill, affecting the assignment. Finally, in our framework, all workers with same cognitive-manual skills matched firms with the same skill demand without randomness. It suggests that we could consider unobserved heterogeneity.

\section{Conclusion}\label{sec: conclusion}

Theoretical matching models often face empirical challenges due to discrepancies between model assumptions and real-world data. \cite{lindenlaub2017} presents a tractable theoretical model suitable for comparative statics and qualitative analysis of multidimensional matching. However, its empirical application is limited by restrictive distributional assumptions on observed characteristics and measurement errors. We generalize this model by relaxing these key distributional restrictions, enabling our models to accommodate datasets with matched pairs, regardless of the underlying characteristic and error distributions. Our simulation results demonstrate the accuracy of our semi-nonparametric estimators across various data-generating processes. Moreover, our flexible models generate significant wage polarization, aligning with U.S. data patterns, whereas the parametric Gaussian model falls short. Our estimated models indicate that task-biased technological progress, favoring cognitive abilities over manual skills, is the primary driver of wage polarization.

Our empirical results are consistent with task-biased technical change raising the relative productivity (and returns) to cognitive skill bundles and contributing to wage polarization through equilibrium sorting. Looking ahead, the diffusion of generative artificial intelligence may further reshape the relative productivity of cognitive tasks and—where AI complements higher-order cognition and coordination—could amplify these forces, although the direction and magnitude remain an empirical question. A first policy implication is to broaden access to cognitive and hybrid skill formation through higher education, STEM pathways, and mid-career retraining targeted to task bundles whose returns are rising. Second, policies that foster broad diffusion of productivity-enhancing technologies and expand access to high-quality job opportunities—across firms, industries, and regions—may mitigate inequality by reducing barriers to adoption and complementarity-driven sorting. Finally, redistributive instruments such as wage insurance and earnings-smoothing policies can help insure workers against technology-driven wage risk during transitions.

Our study opens up promising research avenues. First, incorporating additional dimensions such as interpersonal and digital skills into our model holds great potential. By employing advanced techniques such as artificial neural networks \citep{chen2023efficient} to approximate equilibrium functions in high-dimensional spaces, one could enhance the model's accuracy in elucidating intricate matching patterns and their impact on wage inequality. Second, addressing measurement errors in assessing worker skills is crucial. Overcoming this challenge requires innovative econometric approaches that can accurately estimate models amidst multidimensional measurement errors. Finally, our framework's application extends beyond the worker-job matching problem, offering insights into matching problems in diverse contexts such as the marriage market.

While our approach offers valuable insights, it is essential to acknowledge its limitations. As in \citet{lindenlaub2017}, our frictionless matching model, grounded in optimal transport, assumes that all workers are employed and all jobs are filled. The lack of treatment of unemployed workers and job vacancies (unmatched individuals on both sides of the market) is a limitation of the balanced optimal transport formulation employed here and in \citet{lindenlaub2017}. Incorporating unemployment and vacancies would require extending the framework to partial or unbalanced optimal transport, or to models with search frictions. Introducing randomness in assignment through factors such as search frictions and unobserved heterogeneity poses a fruitful challenge, especially when extending existing one-dimensional theories, e.g., those in \citet{eeckhout2011}, to multidimensional settings.

%\bibliographystyle{ecta-fullname}

%\bibliography{ref}

\newpage
\appendix
\setcounter{page}{1}

\section{Technical proofs}\label{appen: proof}
\begin{proof}[Proof of Proposition \ref{Prop:Matchingiden}]
In equilibrium, the firm maximizes its profit, so the first-order condition of the firm's maximization problem is satisfied:
{\small\[
    \nabla w^{*}\left(x\right)-b=\left.\nabla_{x}x'\tilde{y}\right|_{\tilde{y}=T\left(x\right)}.
\]}
By Theorem 2.12 in \citet{villani2003}, $\left.\nabla_{x}x'\tilde{y}\right|_{\tilde{y}=T\left(x\right)} = \nabla w^{o*}\left(x\right)$ and therefore, 
{\small\[
    \nabla w^{*}\left(x\right)
    =\nabla w^{o*}\left(x\right) + b.
\]}
This implies that $w^{*}\left(x\right)=w^{o*}\left(x\right)+x'b+c$ where $c$ is the constant of integration.
\end{proof}

\begin{proof}[Proof of Theorem \ref{thm:MatchingID}]
We first note from $\mathbb{E}\left[\varepsilon_{wi}|x_{i}\right]=0$ that the convex function $w_{0}^{*}\left(x\right)=\mathbb{E}\left[w_{i}|x_{i}=x\right]=w_{0}\left(x\right)+x'b_{0}$
is identified. Since $\nabla w_{0}^{*}\left(x\right)$ is also identified and $\mathbb{E}\left[y_{i}|x_{i}=x\right]=A_{0}^{-1}\left(\nabla w_{0}^{*}\left(x\right)-b_{0}\right)$, the invertibility of $A_{0}$ implies that $A_{0}$ and $b_{0}$ are identified:
We consider $y_{1},\cdots,y_{d+1}\in\mathcal{Y}$ satisfying Assumption \ref{assu:bilinearID2}. Then, there are corresponding $d+1$ distinct points $x_{1},\cdots,x_{d+1}\in\mathcal{X}$ such that $A_{0}y_{1}=\nabla w_{0}\left(x_{1}\right),\cdots,A_{0}y_{d+1}=\nabla w_{0}\left(x_{d+1}\right)$. It follows from the invertibility of $A_{0}$ that 
$$\nabla w_{0}\left(x_{1}\right)-\nabla w_{0}\left(x_{2}\right)=A_{0}\left(y_{1}-y_{2}\right), \cdots,\nabla w_{0}\left(x_{d}\right)-\nabla w_{0}\left(x_{d+1}\right)=A_{0}\left(y_{d}-y_{d+1}\right)$$ are independent. Since $\nabla w_{0}^{*}\left(x\right)=\nabla w_{0}\left(x\right)+b_{0}$, the differences $\nabla w_{0}^{*}\left(x_{1}\right)-\nabla w_{0}^{*}\left(x_{2}\right),\cdots,\nabla w_{0}^{*}\left(x_{d}\right)-\nabla w_{0}^{*}\left(x_{d+1}\right)$ are also linearly independent. Then, $A_{0}$ is identified from
{\small\[
    \begin{pmatrix}
    \mathbb{E}\left[y_{i}|x_{i}=x_{1}\right]-\mathbb{E}\left[y_{i}|x_{i}=x_{2}\right] \\ \vdots \\
    \mathbb{E}\left[y_{i}|x_{i}=x_{d}\right]-\mathbb{E}\left[y_{i}|x_{i}=x_{d+1}\right]\end{pmatrix}'
    =A_{0}^{-1}\begin{pmatrix}
    \nabla w_{0}^{*}\left(x_{1}\right)-\nabla w_{0}^{*}\left(x_{2}\right) \\ \vdots \\
    \nabla w_{0}^{*}\left(x_{d}\right)-\nabla w_{0}^{*}\left(x_{d+1}\right)\end{pmatrix}',%\\
    %&=A^{-1}\begin{pmatrix}\nabla w_{0}\left(X_{1}\right)-\nabla w_{0}\left(X_{2}\right) &
    %\nabla w_{0}\left(X_{2}\right)-\nabla w_{0}\left(X_{3}\right)\end{pmatrix}.
\]}
In turn, $b_{0}$ and $w_{0}\left(x\right)=w_{0}^{*}\left(x\right)-x'b_{0}$ are also identified.
\end{proof}

\begin{proof}[Proof of Proposition \ref{prop:convrate}]
We obtain the result by applying Theorem 3.2 in \citet{chen2007}. We note that
{\small\[\begin{split}
    &\rho'\left(z;\lambda\right)\Sigma\left(x\right)^{-1}\rho\left(z;\lambda\right)
    -\rho'\left(z;\lambda_{0}\right)\Sigma\left(x\right)^{-1}\rho\left(z;\lambda_{0}\right)\\
    &=\left[\rho\left(z;\lambda\right)-\rho\left(z;\lambda_{0}\right)\right]'
    \Sigma\left(x\right)^{-1}\left[\rho\left(z;\lambda\right)-\rho\left(z;\lambda_{0}\right)\right]
    -2\varepsilon'\Sigma\left(x\right)^{-1}
    \left[\rho\left(z;\lambda\right)-\rho\left(z;\lambda_{0}\right)\right].
\end{split}\]}
and
{\small\[
    \rho\left(z;\lambda\right)-\rho\left(z;\lambda_{0}\right)
    =\begin{pmatrix}
        w\left(x\right)-w_{0}\left(x\right)+x'\left(b-b_{0}\right) \\
        \nabla_{C}w_{0}\left(x\right)\left(\kappa_{C}-\kappa_{C0}\right)
        +\kappa_{C}\nabla_{C}\left\{w\left(x\right)-w_{0}\left(x\right)\right\} \\
        \nabla_{M}w_{0}\left(x\right)\left(\kappa_{M}-\kappa_{M0}\right)
        +\kappa_{M}\nabla_{M}\left\{w\left(x\right)-w_{0}\left(x\right)\right\}
    \end{pmatrix}.
\]}
%\begin{equation*}\begin{split}
%    C^{-1}\left|\rho\left(Z;\lambda\right)-\rho\left(Z;\lambda_{0}\right)\right|_{e}^{2}
%    &\leq\left[\rho\left(Z;\lambda\right)-\rho\left(Z;\lambda_{0}\right)\right]'
%    \Sigma\left(X\right)^{-1}\left[\rho\left(Z;\lambda\right)-\rho\left(Z;\lambda_{0}\right)\right]\\
%    &\leq
%    C\left|\rho\left(Z;\lambda\right)-\rho\left(Z;\lambda_{0}\right)\right|_{e}^{2}.
%\end{split}\end{equation*}
Then, it follows from Assumption \ref{assu:Sigma} that
{\small\[
    \lVert\lambda-\lambda_{0}\rVert^{2}
    \asymp\mathbb{E}\left[\rho'\left(z_{i};\lambda\right)\Sigma\left(x_{i}\right)^{-1}\rho\left(z_{i};\lambda\right)
    -\rho'\left(z_{i};\lambda_{0}\right)\Sigma\left(x_{i}\right)^{-1}\rho\left(z_{i};\lambda_{0}\right)\right],
\]}
i.e., there exists a finite $C_{1}>0$ such that
{\small\[
    C_{1}^{-1}\lVert\lambda-\lambda_{0}\rVert^{2}
    \leq\mathbb{E}\left[\rho'\left(z_{i};\lambda\right)\Sigma\left(x_{i}\right)^{-1}\rho\left(z_{i};\lambda\right)
    -\rho'\left(z_{i};\lambda_{0}\right)\Sigma\left(x_{i}\right)^{-1}\rho\left(z_{i};\lambda_{0}\right)\right]
    \leq C_{1}\lVert\lambda-\lambda_{0}\rVert^{2}.
\]}
Condition 3.6 in \citet{chen2007} is assumed with Assumption \ref{assu:IID}. Now we check Conditions 3.7 and 3.8 in \citet{chen2007}. Again, Assumption \ref{assu:Sigma} implies that there exists a $C_{2}>0$ such that
{\small\[\begin{split}
    &\mathbb{E}
        \left[\left(\rho'\left(z_{i};\lambda_{0}\right)\Sigma\left(x_{i}\right)^{-1}\rho\left(z_{i};\lambda_{0}\right)
            -\rho'\left(z_{i};\lambda\right)\Sigma\left(x_{i}\right)^{-1}\rho\left(z_{i};\lambda\right)\right)^{2}\right]\\
    &\leq C_{2}\mathbb{E}\left[\left|\rho\left(z_{i};\lambda_{0}\right)
                -\rho\left(z_{i};\lambda\right)\right|_{e}^{4}\right].
\end{split}\]}
By Lemma 2 in \citet{chen1998}, we have $\lVert w-w_{0}\rVert_{\infty}\leq c\lVert w-w_{0}\rVert_{2}^{m/\left(m+1\right)}$ and $\lVert\nabla_{C}\left\{w-w_{0}\right\}\rVert_{\infty},\lVert\nabla_{M}\left\{w-w_{0}\right\}\rVert_{\infty}\leq c\lVert w-w_{0}\rVert_{2}^{\left(m-1\right)/m}$ for some finite $c>0$, where $\lVert w\rVert_{2}^{2}=\mathbb{E}\left[w^{2}\left(x_{i}\right)\right]$. Since $\lVert\cdot\rVert_{2}\asymp\lVert\cdot\rVert$,
{\small\[
    \mathbb{E}\left[
    \left(\rho'\left(z_{i};\lambda_{0}\right)\Sigma\left(x_{i}\right)^{-1}\rho\left(z_{i};\lambda_{0}\right)
    -\rho'\left(z_{i};\lambda\right)\Sigma\left(x_{i}\right)^{-1}\rho\left(z_{i};\lambda\right)\right)^{2}
    \right]
    \leq C_{3}\lVert \lambda-\lambda_{0}\rVert^{2\left[1+\left(m-1\right)/m\right]}.
\]}
So Condition 3.7 is satisfied for all $\varepsilon\leq1$. On the other hand,
{\small\[\begin{split}
    &\left|\rho'\left(z_{i};\lambda_{0}\right)\Sigma\left(x_{i}\right)^{-1}\rho\left(z_{i};\lambda_{0}\right)
    -\rho'\left(z_{i};\lambda\right)\Sigma\left(x_{i}\right)^{-1}\rho\left(z_{i};\lambda\right)\right|\\
    &\leq\lVert\lambda-\lambda_{0}\rVert_{\infty}\left|\Sigma\left(x_{i}\right)\right|^{-1}
    \left(2\left|\varepsilon\right|_{e}+\lVert\lambda\rVert_{\infty}+\lVert\lambda_{0}\rVert_{\infty}\right),
\end{split}\]}
almost surely. Using Lemma 2 in \citet{chen1998} again, Condition 3.8 is satisfied.

To apply Theorem 3.2 in \citet{chen2007}, it remains to compute the deterministic approximation error rate $\inf_{\lambda\in\Theta\times\mathcal{W}_{n}}\lVert\lambda-\lambda_{0}\rVert$ and the metric entropy with bracketing. By the same proof as that for Proposition 3.3 in \citet{chen2007}, they are also computed, and then the result follows.
\end{proof}

\begin{proof}[Proof of Theorem \ref{thm:AsympSieveGLSE}]
We obtain the limiting distribution of $\hat{\theta}_{n}$ by verifying that Assumptions 4.1 and 4.2 of Proposition 4.4 in \citet{chen2007} are satisfied. It is easy to see that Assumptions 4.2.(ii) and (iv) are satisfied with the expression for $\rho$. Assumptions 4.1.(i) and 4.2.(i) are our assumption \ref{assu:interior} and \ref{assu:Sigma}. Assumption 4.1.(ii) is implied by our assumption \ref{assu:Sigma} and \ref{assu:Dv*}. Assumption 4.1.(iii) is implied by our Proposition \ref{prop:convrate} and Assumption \ref{assu:Dv*}: there is $\pi_{n}u^{*}\in\mathcal{W}_{n}$ such that $\lVert\pi_{n}u^{*}-u^{*}\rVert\times\lVert\hat{\lambda}_{n}-\lambda_{0}\rVert=o_{p}\left(n^{-1/2}\right)$. Let $u_{\theta}^{*}=\left(u_{\theta1}^{*},u_{\theta2}^{*},u_{\theta3}^{*},u_{\theta4}^{*}\right)'=\left(\mathbb{E}\left[D_{v^{*}}\left(x_{i}\right)'
\Sigma\left(x_{i}\right)^{-1}D_{v^{*}}\left(x_{i}\right)\right]\right)^{-1}\eta$,
$u_{w}^{*}=-v^{*}u_{\theta}^{*}$ and $u^{*}=\left(u_{\theta}^{*},u_{w}^{*}\right)$, where $\eta\in\mathbb{R}^{4}$ is an arbitrary unit vector. It remains to show that Assumption 4.2.(iii) (Conditions 4.2$'$ and 4.3$'$) in \citet{chen2007} are satisfied with
{\small\[\begin{split}
    &\frac{\partial\ell\left(z;\bar{\lambda}\right)}{\partial\lambda}\left[\pi_{n}u^{*}\right]
    =\begin{pmatrix} \bar{w}\left(x\right)+x'\bar{b}-w \\
        \bar{\kappa}_{C}\nabla_{C}\bar{w}\left(x\right)-y_{C} \\
        \bar{\kappa}_{M}\nabla_{M}\bar{w}\left(x\right)-y_{M} \end{pmatrix}'
    \Sigma\left(x\right)^{-1}
    \begin{pmatrix}
        \pi_{n}u_{w}^{*}+u_{\theta3}^{*}x_{C}+u_{\theta4}^{*}x_{M} \\
        u_{\theta1}^{*}\nabla_{C}\bar{w}\left(x\right)
        +\bar{\kappa}_{C}\nabla_{C}\left(\pi_{n}u_{w}^{*}\left(x\right)\right) \\
        u_{\theta2}^{*}\nabla_{M}\bar{w}\left(x\right)
        +\bar{\kappa}_{M}\nabla_{M}\left(\pi_{n}u_{w}^{*}\left(x\right)\right)
    \end{pmatrix}\\
    &=\begin{pmatrix}
        \bar{w}\left(x\right)-w_{0}\left(x\right)+x'\left(\bar{b}-b_{0}\right)-\varepsilon_{w} \\
        \bar{\kappa}_{C}\nabla_{C}\bar{w}\left(x\right)-\kappa_{C0}\nabla_{C}w_{0}\left(x\right)-\varepsilon_{C} \\
        \bar{\kappa}_{M}\nabla_{M}\bar{w}\left(x\right)-\kappa_{M0}\nabla_{M}w_{0}\left(x\right)-\varepsilon_{M}
     \end{pmatrix}'
    \Sigma\left(x\right)^{-1}
    \begin{pmatrix}
        \pi_{n}u_{w}^{*}+u_{\theta3}^{*}x_{C}+u_{\theta4}^{*}x_{M} \\
        u_{\theta1}^{*}\nabla_{C}\bar{w}\left(x\right)
        +\bar{\kappa}_{C}\nabla_{C}\left(\pi_{n}u_{w}^{*}\left(x\right)\right) \\
        u_{\theta2}^{*}\nabla_{M}\bar{w}\left(x\right)
        +\bar{\kappa}_{M}\nabla_{M}\left(\pi_{n}u_{w}^{*}\left(x\right)\right)
    \end{pmatrix}
\end{split}\]}
for all $\bar{\lambda}\in\Theta\times\mathcal{W}_{n}$ with $\lVert\bar{\lambda}-\lambda_{0}\rVert=o\left(1\right)$. Since
{\footnotesize\[
    \begin{pmatrix}
        \hat{w}\left(x\right)-w_{0}\left(x\right)+x'\left(\hat{b}-b_{0}\right) \\
        \hat{\kappa}_{C}\nabla_{C}\hat{w}\left(x\right)-\kappa_{C0}\nabla_{C}w_{0}\left(x\right) \\
        \hat{\kappa}_{M}\nabla_{M}\hat{w}\left(x\right)-\kappa_{M0}\nabla_{M}w_{0}\left(x\right)
    \end{pmatrix}
    =\frac{d\rho\left(z;\lambda_{0}\right)}{d\lambda}\left[\hat{\lambda}_{n}-\lambda_{0}\right]
    +\begin{pmatrix} 0 \\
    \left(\hat{\kappa}_{C}-\kappa_{C0}\right)
    \left(\nabla_{C}\hat{w}\left(x\right)-\nabla_{C}w_{0}\left(x\right)\right) \\
    \left(\hat{\kappa}_{M}-\kappa_{M0}\right)
    \left(\nabla_{M}\hat{w}\left(x\right)-\nabla_{M}w_{0}\left(x\right)\right) \end{pmatrix},
\]}
and
{\small\[\begin{split}
    &\begin{pmatrix}
        \pi_{n}u_{w}^{*}+u_{\theta3}^{*}x_{C}+u_{\theta4}^{*}x_{M} \\
        u_{\theta1}^{*}\nabla_{C}\hat{w}\left(x\right)
        +\hat{\kappa}_{C}\nabla_{C}\left(\pi_{n}u_{w}^{*}\left(x\right)\right) \\
        u_{\theta2}^{*}\nabla_{M}\hat{w}\left(x\right)
        +\hat{\kappa}_{M}\nabla_{M}\left(\pi_{n}u_{w}^{*}\left(x\right)\right)
    \end{pmatrix}\\
    &=\frac{d\rho\left(Z;\lambda_{0}\right)}{d\lambda}\left[\pi_{n}u^{*}\right]
    +\begin{pmatrix} 0 \\ u_{\theta1}^{*}\nabla_{C}\left(\hat{w}\left(x\right)-w_{0}\left(x\right)\right)
    +\left(\hat{\kappa}_{C}-\kappa_{C0}\right)\nabla_{C}\left(\pi_{n}u_{w}^{*}\left(x\right)\right) \\
    u_{\theta2}^{*}\nabla_{M}\left(\hat{w}\left(x\right)-w_{0}\left(x\right)\right)
    +\left(\hat{\kappa}_{M}-\kappa_{M0}\right)\nabla_{M}\left(\pi_{n}u_{w}^{*}\left(x\right)\right)
    \end{pmatrix},
\end{split}\]}
Condition 4.3$'$ is satisfied given the definition of $\lVert\cdot\rVert$ and
{\small\[
    \langle\hat{\lambda}_{n}-\lambda_{0},\pi_{n}u^{*}\rangle
    =\mathbb{E}\left[
    \left(\frac{d\rho\left(z_{i};\lambda_{0}\right)}{d\lambda}\left[\hat{\lambda}_{n}-\lambda_{0}\right]\right)'
    \Sigma\left(x_{i}\right)^{-1}
    \frac{d\rho\left(z_{i};\lambda_{0}\right)}{d\lambda}\left[\pi_{n}u^{*}\right]\right].
\]}
Condition 4.2$'$ can be verified by applying Lemma 4.2 in \citet{chen2007}. The condition on the metric entropy with bracketing for Lemma 4.2 is satisfied with Assumption \ref{assu:Holder}.
\end{proof}

\begin{proof}[Proof of Theorem \ref{thm:EfficSieveGLSE}]
We follow the proofs of Theorem 4.1 and 6.2 in \citet{ai2003}. Let $\mathcal{N}_{on}\equiv\left\{\lambda\in\Theta\times\mathcal{W}_{n}:\lVert\lambda-\lambda_{0}\rVert_{s}=o\left(1\right),\ \lVert\lambda-\lambda_{0}\rVert=o\left(n^{-1/4}\right)\right\}$. By Proposition \ref{prop:convrate}, the sieve estimator $\tilde{\lambda}_{n}$ in Step 1 satisfies
$\lVert\tilde{\lambda}_{n}-\lambda_{0}\rVert_{s}=o_{p}\left(1\right)$ and $\lVert\tilde{\lambda}_{n}-\lambda_{0}\rVert=o_{p}\left(n^{-1/4}\right)$. Hence $\tilde{\lambda}_{n}\in\mathcal{N}_{on}$. Using the proof similar to those of Proposition \ref{prop:convrate}, we can also show that $\hat{\lambda}_{n}\in\mathcal{N}_{on}$. Let $u_{0\theta}=\left(u_{0\theta1},u_{0\theta2},u_{0\theta3},u_{0\theta4}\right)'=\left(\mathbb{E}\left[D_{v_{0}}\left(x_{i}\right)'
\Sigma\left(x_{i}\right)^{-1}D_{v_{0}}\left(x_{i}\right)\right]\right)^{-1}\eta$,
$u_{0w}=-v_{0}u_{0\theta}$ and $u_{0}=\left(u_{0\theta},u_{0w}\right)$, where $\eta\in\mathbb{R}^{4}$ is an arbitrary unit vector. Then, we have
{\small\[
    \left(\theta-\theta_{0}\right)'\eta
    =\langle\lambda-\lambda_{0},u_{0}\rangle
    =\mathbb{E}\left[
    \left(\frac{d\rho\left(z_{i};\lambda_{0}\right)}{d\lambda}\left[\lambda-\lambda_{0}\right]\right)'
    \Sigma_{0}\left(x_{i}\right)^{-1}
    \frac{d\rho\left(z_{i};\lambda_{0}\right)}{d\lambda}\left[u_{0}\right]\right]
\]}
for all $\lambda\in\Lambda$. Let $\varepsilon_{n}=o\left(n^{-1/2}\right)>0$. Denote $u_{n0}:=\left(u_{n0\theta},u_{n0w}\right)=\pi_{n}u_{0}$ to simplify notation. We take a continuous path $\lambda\left(t\right)=\hat{\lambda}_{n}\pm t\varepsilon_{n}u_{n0}$. Then $\left\{\lambda\left(t\right):t\in\left[0,1\right]\right\}$ in $\mathcal{N}_{on}$. Let
{\small\[
    \hat{Q}_{n}\left(\lambda\left(t\right)\right)
    =-\frac{1}{n}\sum_{i=1}^{n}\rho\left(z_{i};\lambda\left(t\right)\right)'
        \hat{\Sigma}_{0}\left(x_{i}\right)^{-1}
        \rho\left(z_{i};\lambda\left(t\right)\right).
\]}
By definition of $\hat{\lambda}_{n}$, and a Taylor expansion around $t=0$ up to second order, we obtain
{\small\[
    0\leq
    \hat{Q}_{n}\left(\hat{\lambda}_{n}\right)-\hat{Q}_{n}\left(\hat{\lambda}_{n}\pm\varepsilon_{n}u_{n0}\right)
    =-\left.\frac{d\hat{Q}_{n}\left(\lambda\left(t\right)\right)}{dt}\right|_{t=0}
    -\frac{1}{2}\left.\frac{d^{2}\hat{Q}_{n}\left(\lambda\left(t\right)\right)}{dt^{2}}\right|_{t=s},
\]}
for some $s\in\left[0,1\right]$, where
{\small\[\begin{split}
    -\left.\frac{d\hat{Q}_{n}\left(\lambda\left(t\right)\right)}{dt}\right|_{t=0}
    =&\frac{\pm2\varepsilon_{n}}{n}\sum_{i=1}^{n}\rho\left(z_{i};\hat{\lambda}_{n}\right)'
        \hat{\Sigma}_{0}\left(x_{i}\right)^{-1}
        \frac{d\rho\left(z_{i};\hat{\lambda}_{n}\right)}{d\lambda}\left[u_{n0}\right]\\
    -\left.\frac{d^{2}\hat{Q}_{n}\left(\lambda\left(t\right)\right)}{dt^{2}}\right|_{t=s}
    =&\underbrace{\frac{2\varepsilon_{n}^{2}}{n}\sum_{i=1}^{n}\rho\left(z_{i};\lambda\left(s\right)\right)'
        \hat{\Sigma}_{0}\left(x_{i}\right)^{-1}
        \frac{d^{2}\rho\left(z_{i};\lambda\left(s\right)\right)}{d\lambda d\lambda}\left[u_{n0},u_{n0}\right]}_{:=A_{1}}\\
    &+\underbrace{\frac{2\varepsilon_{n}^{2}}{n}\sum_{i=1}^{n}
        \left(\frac{d\rho\left(z_{i};\lambda\left(s\right)\right)}{d\lambda}\left[u_{n0}\right]\right)'
        \hat{\Sigma}_{0}\left(x_{i}\right)^{-1}
        \frac{d\rho\left(z_{i};\lambda\left(s\right)\right)}{d\lambda}\left[u_{n0}\right]}_{:=A_{2}},
\end{split}\]}
with
{\small\[
    \frac{d\rho\left(z;\lambda\left(s\right)\right)}{d\lambda}\left[\varepsilon_{n}u_{n0}\right]
    \equiv\left.\frac{d\rho\left(z;\lambda\left(t\right)\right)}{dt}\right|_{t=s},\quad
    \frac{d^{2}\rho\left(z;\lambda\left(s\right)\right)}{d\lambda d\lambda}
        \left[\varepsilon_{n}u_{n0},\varepsilon_{n}u_{n0}\right]
    \equiv\left.\frac{d^{2}\rho\left(z;\lambda\left(t\right)\right)}{dt^{2}}\right|_{t=s}.
\]}
We note from Assumption \ref{assu:Dv*}.(ii) that each element of
{\small\[
    \frac{d^{2}\rho\left(z;\lambda\left(s\right)\right)}{d\lambda d\lambda}\left[u_{n0},u_{n0}\right]
    =\left.\frac{d^{2}\rho\left(z;\lambda+tu_{n0}\right)}{dt^{2}}\right|_{t=s}
    =\begin{pmatrix} 0 \\ -2u_{0\theta1}\nabla_{C}u_{n0w}\left(x\right) \\
    -2u_{0\theta2}\nabla_{M}u_{n0w}\left(x\right) \end{pmatrix}
\]}
is uniformly bounded over $x\in\mathcal{X}$. For $A_{1}$, we write
{\small\[\begin{split}
    &\frac{1}{n}\sum_{i=1}^{n}
        \rho\left(z_{i};\lambda\left(s\right)\right)'\hat{\Sigma}_{0}\left(x_{i}\right)^{-1}
        \frac{d^{2}\rho\left(z_{i};\lambda\left(s\right)\right)}{d\lambda d\lambda}\left[u_{n0},u_{n0}\right]\\
    =&\frac{1}{n}\sum_{i=1}^{n}
        \left(\rho\left(z_{i};\lambda\left(s\right)\right)-\rho\left(z_{i};\lambda_{0}\right)\right)'
        \hat{\Sigma}_{0}\left(x_{i}\right)^{-1}
        \frac{d^{2}\rho\left(z_{i};\lambda\left(s\right)\right)}{d\lambda d\lambda}\left[u_{n0},u_{n0}\right]\\
    &+\frac{1}{n}\sum_{i=1}^{n}
        \rho\left(z_{i};\lambda_{0}\right)'
        \left(\hat{\Sigma}_{0}\left(x_{i}\right)^{-1}-\Sigma_{0}\left(x_{i}\right)^{-1}\right)
        \frac{d^{2}\rho\left(z_{i};\lambda\left(s\right)\right)}{d\lambda d\lambda}\left[u_{n0},u_{n0}\right]\\
    &+\frac{1}{n}\sum_{i=1}^{n}
        \rho\left(z_{i};\lambda_{0}\right)'\Sigma_{0}\left(x_{i}\right)^{-1}
        \frac{d^{2}\rho\left(z_{i};\lambda\left(s\right)\right)}{d\lambda d\lambda}\left[u_{n0},u_{n0}\right].
\end{split}\]}
Since
{\small\[\begin{split}
    \rho\left(z;\lambda\left(s\right)\right)-\rho\left(z;\lambda_{0}\right)
    =&\frac{d\rho\left(z;\lambda_{0}\right)}{d\lambda}\left[\hat{\lambda}_{n}-\lambda_{0}\right]
    +\begin{pmatrix} 0 \\
    \left(\hat{\kappa}_{C}-\kappa_{C0}\right)
    \left(\nabla_{C}\hat{w}\left(x\right)-\nabla_{C}w_{0}\left(x\right)\right) \\
    \left(\hat{\kappa}_{M}-\kappa_{M0}\right)
    \left(\nabla_{M}\hat{w}\left(x\right)-\nabla_{M}w_{0}\left(x\right)\right) \end{pmatrix}\\
    &\mp s\varepsilon_{n}\begin{pmatrix} u_{n0w}\left(x\right)+x_{C}u_{0\theta1}+x_{M}u_{0\theta2} \\
        \hat{\kappa}_{C}\nabla_{C}u_{n0w}\left(x\right)+u_{0\theta3}\nabla_{C}\hat{w}\left(x\right)
        +u_{0\theta3}\nabla_{C}u_{n0w}\left(x\right)\\
        \hat{\kappa}_{M}\nabla_{M}u_{n0w}\left(x\right)+u_{0\theta4}\nabla_{M}\hat{w}\left(x\right)
        +u_{0\theta4}\nabla_{M}u_{n0w}\left(x\right)
    \end{pmatrix},
\end{split}\]}
the first term of the right-hand side is $o_{p}\left(n^{-1/4}\right)$ uniformly over $\lambda\left(s\right)\in\mathcal{N}_{on}$, which implies that $A_{1}$ is $o_{p}\left(\varepsilon_{n}^{2}\right)$. For $A_{2}$, we write
{\small\[\begin{split}
    &\frac{1}{n}\sum_{i=1}^{n}
        \left(\frac{d\rho\left(z_{i};\lambda\left(s\right)\right)}{d\lambda}\left[u_{n0}\right]\right)'
        \hat{\Sigma}_{0}\left(x_{i}\right)^{-1}
        \frac{d\rho\left(z_{i};\lambda\left(s\right)\right)}{d\lambda}\left[u_{n0}\right]\\
    =&\frac{1}{n}\sum_{i=1}^{n}
        \left(\frac{d\rho\left(z_{i};\lambda\left(s\right)\right)}{d\lambda}\left[u_{n0}\right]
            -\frac{d\rho\left(z_{i};\lambda_{0}\right)}{d\lambda}\left[u_{n0}\right]\right)'
        \hat{\Sigma}_{0}\left(x_{i}\right)^{-1}
        \frac{d\rho\left(z_{i};\lambda\left(s\right)\right)}{d\lambda}\left[u_{n0}\right]\\
    &+\frac{1}{n}\sum_{i=1}^{n}
        \left(\frac{d\rho\left(z_{i};\lambda_{0}\right)}{d\lambda}\left[u_{n0}\right]\right)'
        \hat{\Sigma}_{0}\left(x_{i}\right)^{-1}
        \left(\frac{d\rho\left(z_{i};\lambda\left(s\right)\right)}{d\lambda}\left[u_{n0}\right]
            -\frac{d\rho\left(z_{i};\lambda_{0}\right)}{d\lambda}\left[u_{n0}\right]\right)\\
    &+\frac{1}{n}\sum_{i=1}^{n}
        \left(\frac{d\rho\left(z_{i};\lambda_{0}\right)}{d\lambda}\left[u_{n0}\right]\right)'
        \left(\hat{\Sigma}_{0}\left(x_{i}\right)^{-1}-\Sigma_{0}\left(x_{i}\right)^{-1}\right)
        \frac{d\rho\left(z_{i};\lambda_{0}\right)}{d\lambda}\left[u_{n0}\right]\\
    &+\frac{1}{n}\sum_{i=1}^{n}
        \left(\frac{d\rho\left(z_{i};\lambda_{0}\right)}{d\lambda}\left[u_{n0}\right]\right)'
        \Sigma_{0}\left(x_{i}\right)^{-1}
        \frac{d\rho\left(z_{i};\lambda_{0}\right)}{d\lambda}\left[u_{n0}\right]
\end{split}\]}
Since 
{\small\[\begin{split}
    \frac{d\rho\left(z;\lambda\left(s\right)\right)}{d\lambda}\left[u_{n0}\right]
            -\frac{d\rho\left(z;\lambda_{0}\right)}{d\lambda}\left[u_{n0}\right]
    =&\begin{pmatrix} 0 \\ u_{0\theta1}\nabla_{C}\left(\hat{w}\left(x\right)-w_{0}\left(x\right)\right)
    +\left(\hat{\kappa}_{C}-\kappa_{C0}\right)\nabla_{C}u_{n0w}\left(x\right) \\
    u_{0\theta2}\nabla_{M}\left(\hat{w}\left(x\right)-w_{0}\left(x\right)\right)
    +\left(\hat{\kappa}_{M}-\kappa_{M0}\right)\nabla_{M}u_{n0w}\left(x\right)
    \end{pmatrix}\\
    &\pm s\varepsilon_{n}\begin{pmatrix} 0 \\
    u_{0\theta1}\nabla_{C}u_{n0w}\left(x\right)+u_{0\theta1}\nabla_{C}\left(\pi_{n}u_{w}^{*}\left(x\right)\right) \\
    u_{0\theta2}\nabla_{M}u_{n0w}\left(x\right)+u_{0\theta2}\nabla_{M}\left(\pi_{n}u_{w}^{*}\left(x\right)\right)
    \end{pmatrix},
\end{split}\]}
the first two terms on the right-hand side are $o_{p}\left(n^{-1/4}\right)$ and the third term is $o_{p}\left(1\right)$ uniformly over $\lambda\left(s\right)\in\mathcal{N}_{on}$, which implies that $A_{2}$ is $O_{p}\left(\varepsilon_{n}^{2}\right)$. Moreover, since $\varepsilon_{n}=o\left(n^{-1/2}\right)>0$, we obtain uniformly over $\lambda\left(s\right)\in\mathcal{N}_{on}$:
{\small\[
    \frac{1}{n}\sum_{i=1}^{n}\rho\left(z_{i};\hat{\lambda}_{n}\right)'
        \hat{\Sigma}_{0}\left(x_{i}\right)^{-1}
    \frac{d\rho\left(z_{i};\hat{\lambda}_{n}\right)}{d\lambda}\left[u_{n0}\right]
    =o_{p}\left(n^{-1/2}\right).
\]}
Write
{\small\[\begin{split}
    &\frac{1}{n}\sum_{i=1}^{n}\rho\left(z_{i};\hat{\lambda}_{n}\right)'
        \hat{\Sigma}_{0}\left(x_{i}\right)^{-1}
    \frac{d\rho\left(z_{i};\hat{\lambda}_{n}\right)}{d\lambda}\left[u_{n0}\right]\\
    =&\frac{1}{n}\sum_{i=1}^{n}
        \rho\left(z_{i};\hat{\lambda}_{n}\right)'\hat{\Sigma}_{0}\left(x_{i}\right)^{-1}
        \left(\frac{d\rho\left(z_{i};\hat{\lambda}_{n}\right)}{d\lambda}\left[u_{n0}\right]
        -\frac{d\rho\left(z_{i};\lambda_{0}\right)}{d\lambda}\left[u_{0}\right]\right)\\
    &+\frac{1}{n}\sum_{i=1}^{n}\rho\left(z_{i};\hat{\lambda}_{n}\right)'
        \left[\hat{\Sigma}_{0}\left(x_{i}\right)^{-1}
        -\Sigma_{0}\left(x_{i}\right)^{-1}\right]
        \frac{d\rho\left(z_{i};\lambda_{0}\right)}{d\lambda}\left[u_{0}\right]\\
    &+\frac{1}{n}\sum_{i=1}^{n}\rho\left(z_{i};\hat{\lambda}_{n}\right)'
        \Sigma_{0}\left(x_{i}\right)^{-1}
        \frac{d\rho\left(z_{i};\lambda_{0}\right)}{d\lambda}\left[u_{0}\right].
\end{split}\]}
Using proofs for second derivative terms together with Assumption \ref{assu:Sigma2}, the first two terms on the right-hand side are $o_{p}\left(n^{-1/2}\right)$. Since $\left\{\rho\left(z_{i};\lambda\right)'
        \Sigma_{0}\left(x_{i}\right)^{-1}
        \frac{d\rho\left(z_{i};\lambda_{0}\right)}{d\lambda}\left[u_{0}\right]:\lambda\in\mathcal{N}_{on}\right\}$ is a Donsker class,
{\small\[\begin{split}
    &\frac{1}{n}\sum_{i=1}^{n}
        \left(\rho\left(z_{i};\hat{\lambda}_{n}\right)-\rho\left(z_{i};\lambda_{0}\right)\right)'
        \Sigma_{0}\left(x_{i}\right)^{-1}\frac{d\rho\left(z_{i};\lambda_{0}\right)}{d\lambda}\left[u_{0}\right]\\
    &=\mathbb{E}\left[\left(\rho\left(z_{i};\hat{\lambda}_{n}\right)-\rho\left(z_{i};\lambda_{0}\right)\right)'
        \Sigma_{0}\left(x_{i}\right)^{-1}
        \frac{d\rho\left(z_{i};\lambda_{0}\right)}{d\lambda}\left[u_{0}\right]\right]
    +o_{p}\left(n^{-1/2}\right).
\end{split}\]}
With $\rho\left(z_{i};\hat{\lambda}_{n}\right)-\rho\left(z_{i};\lambda_{0}\right)-\frac{d\rho\left(z_{i};\lambda_{0}\right)}{d\lambda}\left[\hat{\lambda}_{n}-\lambda_{0}\right]=o_{p}\left(n^{-1/2}\right)$ uniformly over $x_{i}\in\mathcal{X}$,
{\small\[\begin{split}
    &\frac{1}{n}\sum_{i=1}^{n}\rho\left(z_{i};\hat{\lambda}_{n}\right)'\Sigma_{0}\left(x_{i}\right)^{-1}
        \frac{d\rho\left(z_{i};\lambda_{0}\right)}{d\lambda}\left[u_{0}\right]\\
    &=\langle\hat{\lambda}_{n}-\lambda_{0},u_{0}\rangle
    +\frac{1}{n}\sum_{i=1}^{n}\rho\left(z_{i};\lambda_{0}\right)'\Sigma_{0}\left(x_{i}\right)^{-1}
        \frac{d\rho\left(z_{i};\lambda_{0}\right)}{d\lambda}\left[u_{0}\right]
    +o_{p}\left(n^{-1/2}\right).
\end{split}\]}
Then,
{\small\[
    \sqrt{n}\left(\theta-\theta_{0}\right)'\eta=\sqrt{n}\langle\hat{\lambda}_{n}-\lambda_{0},u_{0}\rangle
    =-\frac{1}{\sqrt{n}}\sum_{i=1}^{n}\rho\left(z_{i};\lambda_{0}\right)'
        \Sigma_{0}\left(x_{i}\right)^{-1}
        \frac{d\rho\left(z_{i};\lambda_{0}\right)}{d\lambda}\left[u_{0}\right]
    +o_{p}\left(1\right),
\]}
and Theorem \ref{thm:EfficSieveGLSE} follows from applying a standard CLT for i.i.d. data.
\end{proof}

\section{Extension to unequal dimensions}\label{appen: unequal}

We extend the identification results in Section~\ref{sec: MatchingModelID} to the case where worker and job characteristics have different dimensions, $d_x\neq d_y$. Without Assumption~\ref{assu:bilinearID} (the invertibility of $A$), we can express the equilibrium as
\[
  w^{*}\left(x\right)=w^{o*}\left(x\right)+x'b+c,\quad
  A\left(y_{1}^{*}\left(x\right),\ldots,y_{d_{y}}^{*}\left(x\right)\right)'
  =\nabla w^{o*}\left(x\right).
\]
Then, the econometric model with measurement errors can be defined by:
\[
\begin{split}
  w_{i}&=w^{*}\left(x_{i}\right)+\varepsilon_{wi}
        =w^{o*}\left(x_{i}\right)+x_{i}'b+c+\varepsilon_{wi},\\
  Ay_{i}&=Ay_{i}^{*}\left(x_i\right)+\varepsilon_{yi}
        =\nabla w^{o*}\left(x_i\right)+\varepsilon_{yi}.
\end{split}
\]
The matrix $A_{0}\in\mathbb{R}^{d_{x}\times d_{y}}$, the vector $b_{0}\in\mathbb{R}^{d_{x}}$, and the convex function $w_{0}\left(\cdot\right)=w^{o*}\left(\cdot\right)+c$ are identified by the following theorem.
\begin{theorem}\label{thm:Unequal}
Let Assumption \ref{assu:Px} hold. Assume that there exist $d_y+1$ values of $x$ such that the $d_y\times d_y$ matrix formed by
differences in the conditional mean of job characteristics,
\[
  M_Y:=\left(\mathbb E[y_i\mid x_i=x_1]-\mathbb E[y_i\mid x_i=x_{d_y+1}],
  \cdots,\mathbb E[y_i\mid x_i=x_{d_y}]-\mathbb E[y_i\mid x_i=x_{d_y+1}]\right)
\]
is invertible. Then, $A_0, b_0,$ and $w_0(\cdot)$ are identified.
\end{theorem}

Intuitively, the additional identification condition (the invertibility of $M_Y$) only requires that the conditional mean of job characteristics varies in $d_y$ linearly independent directions across worker types. In the proof of Theorem~\ref{thm:MatchingID} for the equal-dimension case, Assumption~\ref{assu:bilinearID2} is used to establish such linear independence of the conditional means. In the unequal-dimension case, we impose this rank condition directly by assuming that $M_Y$ is invertible for $d_y+1$ worker types. Hence, the invertibility of $M_Y$ provides a sufficient condition for identification, which can be viewed as a conditional-mean analogue of Assumption~\ref{assu:bilinearID2}.

Note that the equilibrium assignment function is also identified when we have $d_x \ge d_y$ and $A$ has rank $d_y$, as it is recovered by $y^*(x) = (A'A)^{-1}A'\nabla w^{o*}\left(x\right).$ In this case, the twist condition is satisfied and the matching between $x$ and $y$ is pure, which implies the uniqueness of stable matching. When $d_y > d_x,$ the twist condition is not met, so the purity and uniqueness of matching between $x$ and $y$ are not guaranteed.

\begin{proof}[Proof of Theorem \ref{thm:Unequal}]
Similar to the proof of Theorem \ref{thm:MatchingID}, it follows from the exogeneity condition
$\mathbb E[\varepsilon_{wi}\mid x_i]=0$ that the convex function $w_0^*(\cdot)=w_0(x) + x'b_0$
is identified from the conditional expectation of wages. Since $w_0^*$ is convex and continuously
differentiable, its gradient $\nabla w_0^*(x)$ is also identified. Moreover,
\[
\nabla w_0^*(x) = \nabla w_0(x) + b_0\quad
\Rightarrow\quad
\nabla w_0(x) = A_0\,\mathbb E[y_i \mid x_i=x] = \nabla w_0^*(x) - b_0.
\]
Now consider the points $x_1,\dots,x_{d_y+1}$ satisfying the new assumption. For each $k=1,\dots,d_y$,
\[
A_0
\Big(
\mathbb E[y_i \mid x_i=x_k]
-
\mathbb E[y_i \mid x_i=x_{d_y+1}]
\Big)
=
\nabla w_0^*(x_k) - \nabla w_0^*(x_{d_y+1}).
\]

Define the $d_x\times d_y$ matrix
\[
M_X :=
\begin{pmatrix}
\nabla w_0^*(x_1)-\nabla w_0^*(x_{d_y+1}) &
\cdots &
\nabla w_0^*(x_{d_y})-\nabla w_0^*(x_{d_y+1})
\end{pmatrix},
\]
and the $d_y\times d_y$ matrix
\[
M_Y :=
\begin{pmatrix}
\mathbb E[y_i \mid x_i=x_1]-\mathbb E[y_i \mid x_i=x_{d_y+1}] &
\cdots &
\mathbb E[y_i \mid x_i=x_{d_y}]-\mathbb E[y_i \mid x_i=x_{d_y+1}]
\end{pmatrix}.
\]
Stacking the above $d_y$ equalities
gives
\[
A_0 M_Y = M_X.
\]
By the invertibility of $M_Y$, we can solve
uniquely for $A_0$:
\[
A_0 = M_X M_Y^{-1}.
\]
This shows that $A_0 \in \mathbb R^{d_x\times d_y}$ is identified, without any
restriction on the relative sizes of $d_x$ and $d_y$.

Having identified $A_0$, we can recover $b_0$. For any $x$ in the support of $x_i$,
\[
b_0
= \nabla w_0^*(x) - A_0\,\mathbb E[y_i \mid x_i=x],
\]
and the right-hand side is identified from the data. Thus $b_0$ is identified.
This completes the proof.
\end{proof}

\section{Bernstein polynomials with convex constraints}\label{appen: bernstein convexity}
The equilibrium wage function, $w\left(x\right)$, obtained from optimal transport theory is unique and convex (up to a constant). However, finite sample estimators of $w\left(x\right)$, $\hat{w}_{n}\left(x;\gamma\right)$, might be nonconvex at the values close to the boundary of $\mathcal{X}$. To obtain a more stable estimator for $w(x)$, we impose a convexity restriction in the sieve-based estimation procedures without loss of generality. Among many possible linear approximating spaces, we particularly consider the following Bernstein polynomial sieve space:
{\small\begin{multline*}
    \mathcal{W}_{n}
    =\left\{w_{n}:\mathcal{X}\rightarrow\mathbb{R}:
    w_{n}\left(x;\gamma\right)=\sum_{j_{1},\ldots,j_{d}=0}^{k_{n}}\gamma_{j_{1}\cdots j_{d}}
    \left[\prod_{\ell=1}^{d}p_{j_{\ell}}\left(x_{\ell}\right)\right]:\right.\\
    \left.p_{j_{\ell}}\left(x_{\ell}\right)
    =\binom{k_{n}}{j_{\ell}}
    \left(\frac{x_{\ell}-\underline{x}_{\ell}}{\overline{x}_{\ell}-\underline{x}_{\ell}}\right)^{j_{\ell}}
    \left(\frac{\overline{x}_{\ell}-x_{\ell}}{\overline{x}_{\ell}-\underline{x}_{\ell}}\right)^{k_{n}-j_{\ell}}\right\},
\end{multline*}}
for $j_{\ell}=0,1,2,\ldots,k_{n}$ where $p_{j_{\ell}}$ is the Bernstein basis polynomial.

Let $\mathcal{W}^{cvx}$ be the set of midpoint convex functions: 
{\small\[
    \mathcal{W}^{cvx}=
    \left\{w\in C\left(\mathcal{X}\right):
    2w\left(\frac{x_{1}+x_{2}}{2}\right)
    \leq w\left(x_{1}\right)+w\left(x_{2}\right),\
    \forall x_{1},x_{2}\in\mathcal{X}\right\},
\]}
where $C\left(\mathcal{X}\right)$ is the class of all continuous functions on $\mathcal{X}$. We do not assume that the true function $w\left(x\right)$ has derivatives of any order. In fact, $\mathcal{W}^{cvx}$ is the class of all continuous convex functions because a continuous function that is midpoint convex is convex.

First, we consider the one-dimensional ($d=1$) constrained Bernstein polynomial sieve space, $\mathcal{W}_{n}^{cvx}=\left\{w_{n}\left(x;\gamma\right)\in\mathcal{W}_{n}:A\gamma\geq0\right\}$, where
{\small\[
    A\gamma
    \equiv
    \begin{pmatrix}
        1 & -2 & 1 & 0 & \cdots & 0 \\  0 & 1 & -2 & 1 & \cdots & 0\\
        &  & \ddots \\  0 & \cdots & 0 & 1 & -2 & 1
    \end{pmatrix}_{\left(k_{n}-1\right)\times\left(k_{n}+1\right)}
    \begin{pmatrix}
        \gamma_{0} \\ \gamma_{1} \\ \vdots \\ \gamma_{k_{n}}
    \end{pmatrix}
    \geq\begin{pmatrix} 0 \\ 0 \\ \vdots \\ 0 \end{pmatrix}.
\]}
Since the second derivatives of $w_{n}\left(x;\gamma\right)$ can be written
as
{\small\[
    w_{n}^{(2)}\left(x;\gamma\right)
    =\frac{k_{n}\left(k_{n}-1\right)}{\left(\overline{x}-\underline{x}\right)^{2}}
    \sum_{j=0}^{k_{n}-2}
    \left(\gamma_{j+2}-2\gamma_{j+1}+\gamma_{j}\right)
    \begin{pmatrix} k_{n}-2 \\ j \end{pmatrix}
    \left(\frac{x-\underline{x}}{\overline{x}-\underline{x}}\right)^{j}
    \left(\frac{\overline{x}-x}{\overline{x}-\underline{x}}\right)^{k_{n}-2-j},
\]}
the above restriction ensures $w_{n}^{\left(2\right)}\left(\cdot\right)\geq0$
for all $n$. \citet{wang2012} show that $\left\{\mathcal{W}_{n}^{cvx}\right\}$ is nested and dense in $\mathcal{W}^{cvx}$ with respect to sup-norm. 

For the two-dimensional sieve in eq.\eqref{eq: funcspace}, we consider the following linear constraints:
{\small\begin{equation}\label{eq: convex2d}\begin{split}
    &\gamma_{j_{C}+2,j_{M}}-2\gamma_{j_{C}+1,j_{M}}+\gamma_{j_{C},j_{M}}\geq0,\quad
    \forall j_{C}=0,\ldots,k_{C_{n}}-2, j_{M}=0,\ldots,k_{M_{n}},\\
    &\gamma_{j_{C},j_{M}+2}-2\gamma_{j_{C},j_{M}+1}+\gamma_{j_{C},j_{M}}\geq0,\quad
    \forall j_{C}=0,\ldots,k_{C_{n}}, j_{M}=0,\ldots,k_{M_{n}}-2.
\end{split}\end{equation}}
Then, the Bernstein polynomial sieve space with linear constraints \eqref{eq: convex2d} is nested and dense in
{\small\begin{multline*}
    \widetilde{\mathcal{W}}^{cvx}
    =\left\{w\in C\left(\mathcal{X}\right):
    2w\left(\frac{x_{C1}+x_{C2}}{2},x_{M}\right)
    \leq w\left(x_{C1},x_{M}\right)+w\left(x_{C2},x_{M}\right)\ \&\right.\\
    \left.2w\left(x_{C},\frac{x_{M1}+x_{M2}}{2}\right)
    \leq w\left(x_{C},x_{M1}\right)+w\left(x_{C},x_{M2}\right),\right.\\
    \left.\forall\left(x_{C1},x_{M}\right),\left(x_{C2},x_{M}\right),
    \left(x_{C},x_{M1}\right),\left(x_{C},x_{M2}\right)\in\mathcal{X}\right\},
\end{multline*}}
which is larger than
{\small\begin{multline*}
    \mathcal{W}^{cvx}
    =\left\{w\in C\left(\mathcal{X}\right):
    2w\left(\frac{x_{C1}+x_{C2}}{2},\frac{x_{M1}+x_{M2}}{2}\right)
    \leq w\left(x_{C1},x_{M1}\right)+w\left(x_{C2},x_{M2}\right),\right.\\
    \left.\forall\left(x_{C1},x_{M1}\right),\left(x_{C2},x_{M2}\right)\in\mathcal{X}\right\}.
\end{multline*}}

Note that \citet{floater1994} provides sufficient conditions for the two-dimensional Bernstein polynomial $w_{n}\left(x;\gamma\right)$ to be convex, which include linear inequalities \eqref{eq: convex2d} as well as additional nonlinear constraints. We use \eqref{eq: convex2d} for our estimation because (i) they are easy to impose in the optimization procedure and to extend to higher-dimensional functions, and (ii) $w\left(x\right)\in\mathcal{W}^{cvx}\subset\widetilde{\mathcal{W}}^{cvx}$.

\section{Effects of production technology on wage distribution}\label{appen: skewdisp}
In this section, we check whether the \citet{lindenlaub2017} model's predictions on the effects of technological changes on wage distribution hold for the transformed data and other distributions. We consider three DGPs with the same quadratic production function, $s\left(x,y\right)=\alpha_{CC}x_{C}y_{C}+\alpha_{MM}x_{M}y_{M}$, and three different skill distributions, respectively: (1) bivariate normal distribution \citep{lindenlaub2017}, (2) transformed Gumbel copula, and (3) untransformed Gumbel copula.

For the first DGP, we set workers' skill bundle, $x$, and jobs' skill requirements, $y$, to follow standard joint normal distributions with $\rho_{x}=-0.2$ and $\rho_{y}=-0.6$, respectively. For both nonnormal transformed and untransformed DGPs, we set $x=\left(x_{1},x_{2}\right)$ and $y=\left(y_{1},y_{2}\right)$ to follow the Gumbel copula with the shape parameter values 1.25 and 2.5, respectively. Then Kendall's correlation coefficients of $x$ and $y$ are 0.2 and 0.6. For the transformed data, we convert $x$ and $y$ into standard normally distributed variables:
\[
    x_{Ci}=\Phi^{-1}(x_{1i}),\ x_{Mi}=\Phi^{-1}(1-x_{2i}),\
    y_{Cj}=\Phi^{-1}(y_{1j}),\ y_{Mj}=\Phi^{-1}(1-y_{2j}),
\]
so that $(x_{Ci},x_{Mi})$ and $(y_{Cj},y_{Mj})$ are negatively correlated. For the untransformed data,
\[
    x_{Ci}=x_{1i},\ x_{Mi}=1-x_{2i},\ y_{Cj}=y_{1j},\ y_{Mj}=1-y_{2j}.
\]
As mentioned in the main text, there is no guarantee that the transformed data follow a joint normal distribution, implying that no closed-form solution exists for the equilibrium wage and assignment function. Hence, we numerically solve the equilibrium matching through linear programming for each Monte Carlo sample and different parameter values of $\left(\alpha_{CC},\alpha_{MM}\right)$.

Figure \ref{fig:wageskewdisp} plots the skewness and variance of the distribution of wages for each DGP. As \citet{lindenlaub2017} derived, for all three DGPs, wage distributions are positively skewed for different pairs of $\alpha_{CC}$ and $\alpha_{MM}$, and the wage dispersion increases as cognitive or manual skill complementarity increases. However, for the transformed Gumbel copula, the skewness decreases as $\alpha_{CC}$ decreases and $\alpha_{MM}$ increases. This simulation result is inconsistent with the theoretical result for normally distributed $X$ and $Y$, in which the skewness is minimized when $\alpha_{CC}=\alpha_{MM}$. It implies that estimating the Gaussian model in \citet{lindenlaub2017} with transformed data can mislead the effects of technological changes on wages and inequality.

\begin{figure}[t!]
\begin{centering}
\includegraphics[width=1\textwidth]{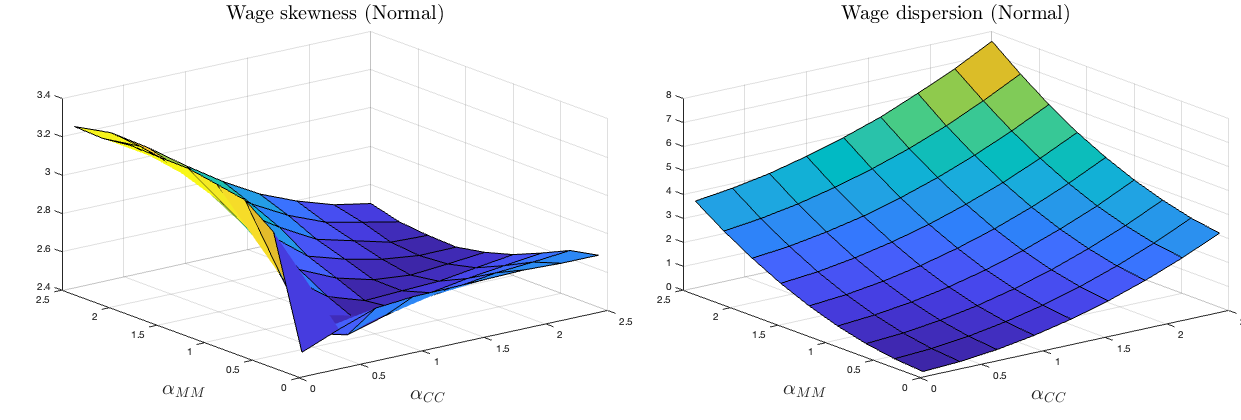}
\medskip
\includegraphics[width=1\textwidth]{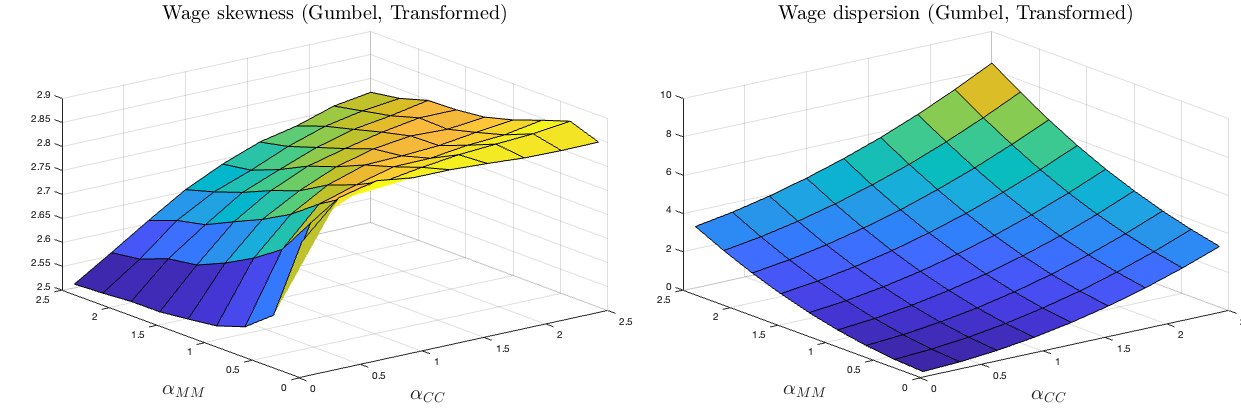}
\medskip
\includegraphics[width=1\textwidth]{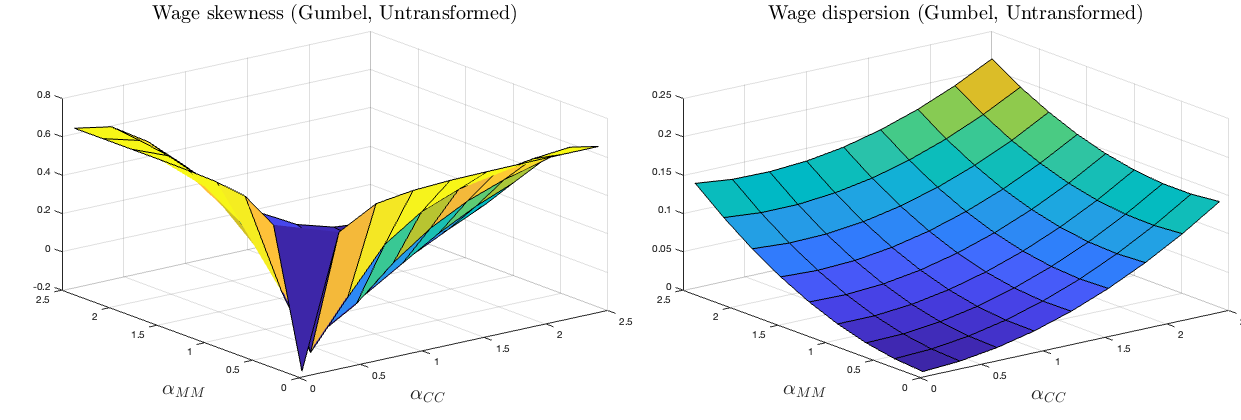}
\par\end{centering}
\caption{\label{fig:wageskewdisp}Effects of changes in production technology on wages and inequality}
\end{figure}

\end{document}